%% file: main.tex
\documentclass[sigconf, nonacm]{acmart}
\usepackage{subfig}
\usepackage[font=small,labelfont=bf]{caption}
\usepackage[skip=4pt]{caption}

\usepackage{balance}
\usepackage{amsmath,amsfonts}
\usepackage{algorithm}
\usepackage[noend]{algpseudocode}
\usepackage{graphicx}
\usepackage{textcomp}
\usepackage{xcolor}
\usepackage[justification=centering]{caption}
\usepackage{setspace}
\usepackage{multirow}
\usepackage{multicol}
\usepackage[inline]{enumitem}
\usepackage{enumitem}
\usepackage[colorinlistoftodos]{todonotes}
\usepackage{color-edits}
\addauthor[Kartik]{kl}{black}
\addauthor[RajKannan]{rk}{green}
\addauthor[Cesar]{cafdr}{red}
\usepackage{etoolbox}
\usepackage{commath}

\usepackage{microtype}
\usepackage{bbm}

\algnewcommand\algorithmicparfor{\textbf{for}}
\algnewcommand\algorithmicpardo{\textbf{do in parallel}}
\algnewcommand\algorithmicforeach{\textbf{for each}}
\algnewcommand{\IfThenElse}[3]{% \IfThenElse{<if>}{<then>}{<else>}
  \State \algorithmicif\ #1\ \algorithmicthen\ #2\ \algorithmicelse\ #3}

\algrenewtext{ParFor}[1]{\algorithmicparfor\ #1\ \algorithmicpardo}
\algrenewtext{EndParFor}{}%{\algorithmicendparfor}

\algrenewtext{ParForEach}[1]{\algorithmicforeach\ #1\ \algorithmicpardo}
\algrenewtext{EndParForEach}{}%{\algorithmicendparfor}

\algrenewtext{ForEach}[1]{\algorithmicforeach\ #1\ \algorithmicfordo}
\algrenewtext{EndForEach}{}
\algdef{SnE}[PARFOR]{ParFor}{EndParFor}[1]{\algorithmicparfor\ #1\ \algorithmicpardo}%
\algdef{SnE}[PARFOREACH]{ParForEach}{EndParForEach}[1]{\algorithmicforeach\ #1\ \algorithmicpardo}%

\algdef{SnE}[FOREACH]{ForEach}{EndForEach}[1]{\algorithmicforeach\ #1}%
\algdef{SE}[DOWHILE]{Do}{doWhile}{\algorithmicdo}[1]{\algorithmicwhile\ #1}%
\algdef{SnE}[WHILEWODO]{Whilewodo}{EndWhilewodo}[1]{\algorithmicwhile\ #1 \{\}}%

\makeatletter
\renewcommand{\Function}[2]{%
  \csname ALG@cmd@\ALG@L @Function\endcsname{#1}{#2}%
  \def\jayden@currentfunction{#1}%
}
\newcommand{\funclabel}[1]{%
  \@bsphack
  \protected@write\@auxout{}{%
    \string\newlabel{#1}{{\jayden@currentfunction}{\thepage}}%
  }%
  \@esphack
}
\makeatother

\newtheorem{theorem}{Theorem}
\newtheorem{lemma}{Lemma}

\newtheorem{definition}{Definition}

\newcommand{\beq}{\begin{equation}}
\newcommand{\eeq}{\end{equation}}
\newcommand{\bea}{\begin{eqnarray}}
\newcommand{\eea}{\end{eqnarray}}

\newcommand\vldbdoi{XX.XX/XXX.XX}

\begin{document}

\title{RECEIPT: REfine CoarsE-grained IndePendent Tasks for \\Parallel Tip decomposition of Bipartite Graphs}

 %  in this sample file, there are a *total*
% of EIGHT authors. SIX appear on the 'first-page' (for formatting
% reasons) and the remaining two appear in the \additionalauthors section.

\author{Kartik Lakhotia, Rajgopal Kannan, Viktor Prasanna}
\affiliation{%
  \institution{Ming Hsieh Department of Electrical Engineering\\ University of Southern California}
}
\email{{klakhoti, rajgopak, prasanna}@usc.edu}

\author{Cesar A. F De Rose}
\affiliation{%
  \institution{School of Technology\\ Pontifical Catholic University of Rio Grande do Sul}
}
\email{cesar.derose@pucrs.br}

% \author{Valerie B\'eranger}
% \orcid{0000-0001-5109-3700}
% \affiliation{%
%   \institution{Inria Paris-Rocquencourt}
%   \city{Rocquencourt}
%   \country{France}
% }
% \email{vb@rocquencourt.com}

% \author{J\"org von \"Arbach}
% \affiliation{%
%   \institution{University of T\"ubingen}
%   \city{T\"ubingen}
%   \country{Germany}
% }
% \email{jaerbach@uni-tuebingen.edu}
% \email{myprivate@email.com}
% \email{second@affiliation.mail}

% \author{Wang Xiu Ying}
% \author{Zhe Zuo}
% \affiliation{%
%   \institution{East China Normal University}
%   \city{Shanghai}
%   \country{China}
% }
% \email{firstname.lastname@ecnu.edu.cn}

% \author{Donald Fauntleroy Duck}
% \affiliation{%
%   \institution{Scientific Writing Academy}
%   \city{Duckburg}
%   \country{Calisota}
% }
% \affiliation{%
%   \institution{Donald's Second Affiliation}
%   \city{City}
%   \country{country}
% }
% \email{donald@swa.edu}

\begin{abstract}
Tip decomposition is a crucial kernel for mining dense subgraphs in bipartite networks, with applications in spam detection, analysis of affiliation networks etc.
% Tip decomposition is a crucial kernel for dense subgraph mining in bipartite networks. 
It creates a hierarchy of vertex-induced subgraphs with varying densities determined by the participation of vertices in butterflies ($2,2-$bicliques).
%of varying densities, where density is identified by the number of  butterflies (2,2-bicliques) incident on vertices.
% Butterflies (2,2-bicliques) are the basic unit of cohesion in bipartite graphs. Tip-decomposition builds a hierarchy of vertex induced subgraphs based on the number of butterflies incident on vertices in those subgraphs. It is a crucial analytic used for dense subgraph mining in bipartite networks. 
%However, tip-decomposition is computationally very expensive.
% Existing decomposition algorithms iteratively create the levels of hierarchy in increasing order of subgraph density. At each level, vertices with minimum butterflies are deleted (peeled) and butterfly count of their 2-hop neighbors is updated. 
% The 2-hop neighborhood exploration renders tip-decomposition 
% computationally very expensive.
% To build the hierarchy, existing algorithms iteratively delete (peel) vertices with minimum number of butterflies and update butterfly count of their 2-hop neighbors. 
%To build the hierarchy, existing algorithms iteratively delete (peel) vertices with minimum number of butterflies, which in turn affects the butterfly count of their 2-hop neighbors. The need to explore 2-hop neighborhood renders tip-decomposition computationally very expensive. Furthermore, the compulsion to peel only minimum butterfly vertices imposes  sequential restrictions on the order of vertex peeling, rendering naive parallelism ineffective and prone to heavy synchronization. 
To build the hierarchy, existing algorithms iteratively follow a {\it delete-update}(peeling) process: deleting  vertices with the minimum number of butterflies and correspondingly updating the butterfly count of their 2-hop neighbors. The need to explore 2-hop neighborhood renders tip-decomposition computationally very expensive. Furthermore, the inherent sequentiality in peeling only minimum butterfly vertices makes derived parallel algorithms prone to heavy synchronization. 
% Furthermore, the strict order of vertex peeling imposes sequential restrictions, rendering naive parallelism ineffective and prone to heavy synchronization. 

In this paper, we propose a novel parallel tip-decomposition algorithm -- REfine CoarsE-grained Independent Tasks (RECEIPT) that relaxes the peeling order restrictions by partitioning the vertices into multiple independent subsets that can be concurrently peeled. This enables RECEIPT to simultaneously achieve a high degree of parallelism and dramatic reduction in synchronizations.
%In this paper, we propose a novel parallel tip-decomposition algorithm -- REfine CoarsE-grained Independent Tasks (RECEIPT). 
% RECEIPT partitions the vertices into multiple independent subsets that can be concurrently peeled. It avoids the ordering restrictions on vertex peeling to simultaneously achieve high degree of parallelism and reduced synchronization.
% Further, the pre-processing step of RECEIPT incurs very few synchronizations and can be efficiently parallelized. 
% RECEIPT uses a pre-processing step to partition the vertices into multiple independent subsets that can be concurrently peeled. Further, the pre-processing step of RECEIPT incurs very few synchronizations and can be efficiently parallelized. 
Further, RECEIPT employs a hybrid peeling strategy along with other optimizations that drastically reduce the amount of wedge exploration and execution time.

%It further employs several optimizations that enable both theoretical work-efficiency and practical performance. 

We perform detailed experimental evaluation of RECEIPT on a shared-memory multicore server. It can process some of the largest publicly available bipartite datasets \textit{orders of magnitude faster} than the state-of-the-art algorithms -- achieving up to $1100\times$ and $64\times$ reduction in the number of thread synchronizations and traversed wedges, respectively. Using $36$ threads, RECEIPT can provide up to $17.1\times$ self-relative speedup. \kledit{Our implementation of RECEIPT is available at \url{https://github.com/kartiklakhotia/RECEIPT}.}

% \begin{itemize}
%     \item Dense subgraph mining for bipartite graphs is important
%     \item Butterflies basic unit of cohesion in bipartite graphs. Sariyuce et al. 
%     define the notion of k-tips/tip number to find vertex-induced subgraphs with many butterflies.
%     \item However, finding k-tips or tip values of vertices is computationally expensive.
%     \item Sequential dependency between multiple rounds of peeling, conventional approaches of
%     parallelizing each round suffer from poor scalability and heavy synchronization.
%     \item We propose the REfine CoarsE-grained IndePendent Tasks (RECEIPT) algorithm that drastically reduces synchronization by partitioning the vertices into multiple independent subsets, each of which can be processed independently and concurrently with the others. RECEIPT further ensures that partitioning incurs few synchronizations and can be done efficiently in parallel. 
    
%     Enabled with several optimizations, RECEIPT delivers both theoretical work-efficiency and practical performance. We perform detailed experimental evaluation of our approach and show that RECEIPT can process some of the largest publicly available bipartite datasets with orders of magnitude reduction in execution time and synchronization of parallel workers. On $36$ threads, RECEIPT can achieve up to $A\times$ self-relative speedup.
% \end{itemize}
\end{abstract}

\maketitle
% \begingroup\small\noindent\raggedright\textbf{PVLDB Reference Format:}\\
% \vldbauthors. \vldbtitle. PVLDB, \vldbvolume(\vldbissue): \vldbpages, \vldbyear.\\
% \href{https://doi.org/\vldbdoi}{doi:\vldbdoi}
% \endgroup
% \begingroup
% \renewcommand\thefootnote{}\footnote{\noindent
% This work is licensed under the Creative Commons BY-NC-ND 4.0 International License. Visit \url{https://creativecommons.org/licenses/by-nc-nd/4.0/} to view a copy of this license. For any use beyond those covered by this license, obtain permission by emailing \href{mailto:info@vldb.org}{info@vldb.org}. Copyright is held by the owner/author(s). Publication rights licensed to the VLDB Endowment. \\
% \raggedright Proceedings of the VLDB Endowment, Vol. \vldbvolume, No. \vldbissue\ %
% ISSN 2150-8097. \\
% \href{https://doi.org/\vldbdoi}{doi:\vldbdoi} \\
% }\addtocounter{footnote}{-1}\endgroup
%%% VLDB block end %%%
 \input{introduction}
\input{background}
\input{receipt}
\input{optimizations}
\input{experiments}

\input{related}
\input{conclusion}

% %ACKNOWLEDGMENTS are optional
% \section{Acknowledgments}
% This section is optional; it is a location for you
% to acknowledge grants, funding, editing assistance and
% what have you.  In the present case, for example, the
% authors would like to thank Gerald Murray of ACM for
% his help in codifying this \textit{Author's Guide}
% and the \textbf{.cls} and \textbf{.tex} files that it describes.

\clearpage\newpage
\bibliographystyle{abbrv}
\bibliography{bibliography}  % vldb_sample.bib is the name of the Bibliography in this case
% You must have a proper ".bib" file
%  and remember to run:
% latex bibtex latex latex
% to resolve all references

% \clearpage
% \newpage
% \input{rebuttal}

\end{document}

%% file: introduction.tex
\section{Introduction}\label{sec:introduction}
Dense subgraph mining is a fundamental problem used for anomaly detection, spam filtering, social network analysis, trend summarizing and several other real-world applications~\cite{anomalyDet, spamDet, communityDet, yang2018mining, otherapp1, otherapp2}. Many of these modern day applications use bipartite graphs to effectively represent two-way relationship structures, such as consumer-product purchase history~\cite{consumerProduct}, user-group memberships in social networks~\cite{orkut}, author-paper networks~\cite{authorPaper}, etc. Consequently, mining cohesive structures in bipartite graphs has gained tremendous interest in recent years~\cite{wangButterfly, wangBitruss, zouBitruss, sariyucePeeling, shiParbutterfly}.
%bipartite networks has gained a lot of interest in is widely used in real-world bipartite graphs for anomaly detection~\cite{anomalyDet}, spam filtering, 

Many techniques have been developed to uncover hierarchical dense structures in unipartite graphs, such as truss and core decomposition~\cite{spamDet, graphChallenge, trussVLDB, sariyuce2016fast, coreVLDB, coreVLDBJ, bonchi2019distance,wen2018efficient}. Such off-the-shelf analytics can be conveniently utilized for discovering dense parts in projections of bipartite graphs as well~\cite{projection}. However, this approach results in a loss of information and a blowup in the size of the projection graphs~\cite{sariyucePeeling}. To prevent these issues, researchers have explored the role of butterflies ($2,2-$bicliques) to mine dense subgraphs directly in bipartite networks~\cite{sariyucePeeling,zouBitruss}. 
A butterfly is the most basic unit of cohesion in bipartite graphs. Recently, Sariyuce et al. conceptualized $k-$tip as a vertex-induced subgraph with at least $k$ butterflies incident on every vertex in one of the bipartite vertex sets~(fig.\ref{fig:tipDemo}).
% based on the incidence of butterflies on vertices in dense subgraphs\cite{sariyucePeeling}.
They show that $k-$tips can unveil hierarchical dense regions in bipartite graphs more effectively than unipartite approaches applied on projection graphs. 
As a space-efficient representation for the $k-$tip hierarchy, Sariyuce et al. further define the notion of \textit{tip number} of a vertex $u$ as the largest $k$ for which a $k-$tip contains $u$. In this paper, we study the problem of finding tip numbers of vertices in a bipartite graph, also known as \textit{tip decomposition}.  

\begin{figure}[htbp]
    \centering
\includegraphics[width=\linewidth]{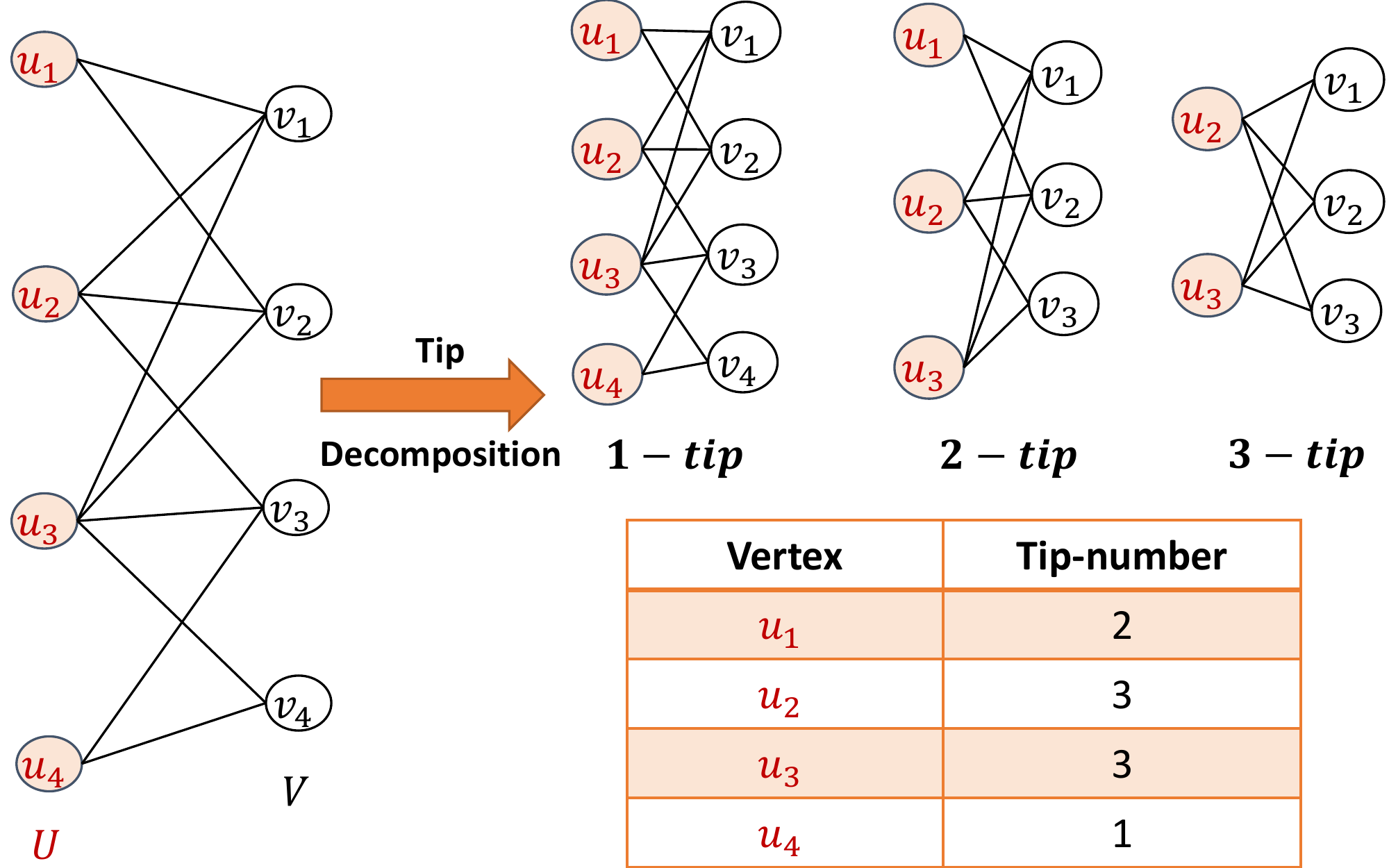}     
\caption{Tip decomposition of vertex set $U$ in a bipartite graph. $u_4$ and $u_1$ participate in $1$ and $2$ butterflies, respectively. Although $u_3$ participates in $5$ butterflies in original graph, only $3$ of them are with $u_2$ with which it creates a $3$-tip.}
    \label{fig:tipDemo}
\end{figure}

Tip decomposition can be employed in several real-world applications that utilize dense subgraphs. It can find groups of researchers (along with group hierarchies) with common affiliations from author-paper networks~\cite{sariyucePeeling}. It can be used to detect communities of spam reviewers from user-rating graphs in databases of e-commerce companies; such reviewers collaboratively rate selected products, appearing in close-knit structures~\cite{mukherjee2012spotting,fei2013exploiting} that tip decomposition can unveil. It can be used for document clustering, graph summarization and link prediction as dense $k$-tips unveil groups of nodes with connections to common and similar sets of neighbors ~\cite{dhillon2001co,leicht2006vertex,navlakha2008graph,communityDet}.

Existing sequential and tip decomposition algorithms~\cite{sariyucePeeling, shiParbutterfly} employ a bottom-up peeling approach. 
%They iteratively peel (delete) vertices with minimum butterfly count and update the count of their $2-$hop neighbors with shared butterflies. 
Vertices with the minimum butterfly count are peeled (deleted), the count of their 2-hop neighbors with shared butterflies is decremented, and the process is then iterated.
However, exploring $2-$hop neighborhood for every vertex requires traversing a large number of wedges, rendering tip decomposition computationally intensive. For example, the $TrU$ dataset has $140$ million edges but peeling it requires traversing $211$ trillion  wedges, which is intractable for a sequential algorithm.
%Tip numbers are assigned to vertices in a non-decreasing order depending on the latest $k-$tip in which the vertices participate before being deleted. 
%However, this method suffers from a major drawback: Support updation requires 
%traversing all wedges of vertices deleted in an iteration which transforms to 
%traversing all wedges in the graph twice during the course of execution. Due to the quadratic growth of wedges wrt vertex degrees, it can be computationally infeasible to process large graphs. For example, tip decomposition of $TrU$ dataset requires exploring $411$ trillion wedges and several days of execution time (table~\ref{table:trackersU}).
For such high complexity analytics, parallel computing is often used to scale to large datasets~\cite{Park_2016,smith2017truss,10.1145/3299869.3319877}. In case of tip decomposition, the bottom-up peeling approach used by the existing parallel algorithms severely restricts their scalability. It mandates constructing levels of the $k-$tip hierarchy in non-decreasing order of tip numbers, thus imposing sequential restrictions on the order of vertex peeling. Note that the parallel threads need to synchronize at the end of every peeling iteration.
%Further, the large range of tip numbers requires several peeling iterations\footnote{Note that vertices with a given tip-number may take multiple iterations to peel.}. 
%%It further mandates a partial order of vertex peeling\footnote{In an iteration, if multiple vertices have minimum support, they can be peeled in any order.}, thus imposing sequential restrictions. 
Further note that the range of tip numbers can be quite large and the number of iterations required to peel all vertices with a given tip number 
is variable. 
% Moreover, the impact  of peeling minimum butterfly count vertices (in terms of number of neighbors to be updated) is also variable. 
%As a result, 
Taken together, the conventional approach of parallelizing the workload within each iteration requires major synchronization, rendering it ineffective. 
For example, P\textsc{ar}B\textsc{utterfly} experiences $\approx1.5$ million synchronization rounds for peeling $TrU$~\cite{shiParbutterfly}. These observations motivate the need for an algorithm that exploits the parallelism available across multiple levels of a $k-$tip hierarchy to reduce the amount of synchronizations. 

In this paper, we propose the REfine CoarsE-grained IndePendent Tasks (RECEIPT) algorithm, that adopts a novel two-step approach to drastically reduce the amount of parallel peeling iterations and in turn, the amount of synchronization.  
% into into the parallelism across the levels of $k-$tip hierarchy.
% These observations motivate the need of an algorithm that exploits the parallelism not just within, but across the levels of $k-$tip hierarchy. In this paper, we propose the two-step REfine CoarsE-grained IndePendent Tasks (RECEIPT) algorithm that taps into this potential parallelism opportunity. 
The key insight that drives the development of RECEIPT is that the tip-number $\theta_u$ of a vertex $u$ only depends on the butterflies shared between $u$ and other vertices with tip numbers \textit{greater than} or \textit{equal to} $\theta_u$. Thus, vertices with smaller tip numbers can be peeled in any order without affecting the correctness of $\theta_u$ computation.
% The key insight that drives the development of RECEIPT is that the tip-number $\theta_u$ of a vertex $u$ only depends on the butterflies $u$ shares with other vertices having tip numbers \textit{equal to $\theta_u$} and greater than $:
% \begin{enumerate}[leftmargin=*]
% \itemsep0em
% \item The cumulative effect of peeling all the vertices with tip-number \textit{less than} $\theta_u$, and not on their order. 
% \item The effect of peeling vertices with tip-number \textit{equal to} $\theta_u$ in the partial order determined by sequential peeling.
% \end{enumerate}
    %<TODO> the following sentence is confusing, doesn't clarify how multiple subsets are created.
To this purpose, RECEIPT divides tip decomposition into a \textit{two-step computation} where each step exposes parallelism across a different dimension. 

In the first step, it creates few \textit{non-overlapping} ranges that represent lower and upper bounds on the vertices' tip numbers. To find vertices with a lower bound $\theta$, all vertices with upper bound smaller than $\theta$ can be peeled simultaneously, providing \textit{sufficient  vertex-workload per parallel peeling iteration}. The small number of ranges ensures \textit{little synchronization} in this step.
% This is achieved by partitioning the spectrum of tip numbers into few \textit{non-overlapping} ranges that represent the bounds. 
% In the first step, it  partitions the  entire set of tip-numbers into few \textit{non-overlapping} ranges and for each vertex, finds the range to which its tip number belongs. The relatively small number of ranges (compared to the set of distinct tip-numbers) ensures \textit{little synchronization} and \textit{sufficient parallel workload per peeling iteration} in this step.
% The second step of RECEIPT \textit{concurrently processes vertex subsets of multiple ranges} to compute the exact tip numbers\footnote{A vertex subset for a given range is peeled by a single thread in the second step.}.
% The absence of overlap between tip number ranges enables the subset of vertices corresponding to a range to be peeled \textit{independently} of other vertices in the graph.  
In the second step, RECEIPT \textit{concurrently processes  vertex subsets} corresponding to different ranges to compute the exact tip numbers\footnote{A vertex subset for a given range is peeled by a single thread in the second step.}.
The absence of overlap between tip number ranges enables each of these subsets to be peeled \textit{independently} of other vertices in the \looseness=-1graph.

RECEIPT's two-step approach further enables development of novel optimizations that radically decrease the amount of wedge exploration.
Equipped with these optimizations, we combine both \textit{computational efficiency} and \textit{parallel performance} for fast decomposition of large bipartite graphs. Overall, our contributions can be summarized as follows:

\begin{enumerate}[leftmargin=*]
    \item We propose a novel RefinE CoarsE-grained Independent Tasks (RECEIPT) algorithm for tip decomposition in bipartite graphs. RECEIPT is the \textit{first} algorithm that parallelizes workload across the levels of subgraph hierarchy created by tip decomposition. 
    \item We show that RECEIPT is theoretically work efficient and dramatically reduces thread synchronization. As an example, it incurs only $1335$ synchronization rounds while processing $TrU$, which is $1105\times$ less than the existing parallel algorithms. 
    \item We develop novel optimizations enabled by the two-step approach of RECEIPT. These optimizations drastically reduce the amount of wedge exploration and improve computational efficiency of RECEIPT. For instance, we traverse  $3297$B wedges to tip decompose $TrU$, which is $64\times$ less than the state-of-the-art.
\end{enumerate}

We conduct detailed experiments using some of the largest public real-world bipartite graphs. RECEIPT extends the limits of current practice by feasibly computing tip decomposition for these datasets. For example, it can process the $TrU$ graph in $46$ minutes whereas the state-of-the-art does not finish in $10$ days. Using $36$ threads, we achieve up to $17.1\times$ parallel speedup.

%% file: background.tex
\section{Background}\label{sec:background}
In this section, we will discuss state-of-the-art algorithms for counting per-vertex butterflies (sec.\ref{sec:counting}). Counting is used to initialize the \textit{support} (running count of incident butterflies) of each vertex during tip-decomposition, and also constitutes a crucial optimization in RECEIPT. Hence, it is imperative to analyze counting algorithms.

We will also discuss the bottom-up peeling approach for tip-decomposition used by the 
existing algorithms (sec.\ref{sec:bottomup}). Table~\ref{table:notations} lists the 
notations used in this paper. Note that we decompose either $U$ or $V$ vertex set at a time.
For clarity of description, we assume that $U$ is the  primary vertex set to process \kledit{and we use the word "wedge" to imply a wedge with endpoints in $U$}.
However, for empirical analysis, we will individually tip decompose both vertex sets in each graph dataset.

\begin{table}[htbp]
\caption{Frequently used notations}
\label{table:notations}
\resizebox{\linewidth}{!}{%
\begin{tabular}{|c|c|}
\hline
$G(W=(U, V), E)$                                 & \begin{tabular}[c]{@{}c@{}}bipartite graph $G$ with vertices $W$ and edges $E$\\ $U$ and $V$ are two disjoint vertex sets in $G$\end{tabular} \\ \hline
$n, m$                                           & no. of  vertices and edges in G; $n=|W|, m = |E|$                                                                                               \\ \hline
$\alpha$ & arboricity of $G$ \cite{chibaArboricity} \\ \hline
$N_u$                                            & neighboring vertices of $u$                                                                                                                   \\ \hline
$d_u$                                            & degree of vertex $u$; $d_u=\abs{N_u}$                                                                                                            \\ \hline
$\bowtie_u$ / $\bowtie_U$                        & support (\# butterflies) of $u$ / vertices in set $U$                                                                                       \\ \hline
% $\bowtie(u_1, v_1, u_2, v_2)$ & butterfly containing $(u_1, u_2)\in U$ and $(v_1, v_2)\in V$                                                                       \\ \hline
$\bowtie_{u_1, u_2}=\ \bowtie_{u_2, u_1}$ & \# butterflies shared between $u_1$ and $u_2$ \\ \hline
$\wedge_U$ & \# wedges with endpoints in set $U$ \\ \hline
$\theta_u$                                       & tip number of vertex $u$                                                                                                                      \\ \hline
$\theta^{max}$ & maximum tip number for a vertex \\ \hline
$P$                                                & number of vertex subsets created by RECEIPT                                                                                    \\ \hline
$T$                                                & number of threads                                                                                                                             \\ \hline
\end{tabular}}
\end{table}

%Any further description of graph/problem.
\subsection{Per-vertex butterfly counting}\label{sec:counting}
A butterfly (2,2-bicliques/quadrangle) is a combination of two wedges with common endpoints. 
% are the smallest unit of cohesion in bipartite graphs, analogous to triangles for unipartite graphs. They 
% are the smallest non-trivial motif in bipartite graphs. Due to their importance in bipartite analytics, researchers have developed efficient algorithms for butterfly counting~\cite{}.
A simple way to count butterflies is to explore all wedges and combine the ones with common end points. However, counting per-vertex butterflies using this procedure is extremely expensive with $\mathcal{O}\left(\sum_{u\in U}\sum_{v\in N_u}d_v\right)$ complexity (if we use vertices in $U$ as end points). 

Chiba and Nishizeki \cite{chibaArboricity} proposed an efficient vertex-priority quadrangle counting algorithm which \kledit{traverses $\mathcal{O}\left(\sum_{(u,v)\in E}\min{(d_u, d_v)}\right) = \mathcal{O}\left(\alpha \cdot m\right)$ wedges with $\mathcal{O}(1)$ work per wedge}. Wang et al.\cite{wangButterfly} further propose a cache efficient version of the vertex-priority algorithm that reorders the vertices in decreasing order of degree. \kledit{Their algorithm only traverses wedges where degree of the last vertex is greater than the degree of the start and middle vertices}. 
%\kldelete{with end-points (in both $U$ and $V$) having higher degree than the start and the mid-points. }
It can be easily parallelized by concurrently processing multiple start vertices~\cite{shiParbutterfly, wangButterfly}, as shown in alg.\ref{alg:counting}. 
In RECEIPT, we adopt this parallel variant for per-vertex counting by adding the contributions from traversed wedges to start, mid and end-points. Each thread is provided a $\mathcal{\theta}(|W|)$ array for wedge aggregation (line 5, alg.\ref{alg:counting}). This is similar to the best performing batch aggregation mode in the P\textsc{ar}B\textsc{utterfly} framework~\cite{shiParbutterfly}.

% \begin{itemize}
%     \item Required to initialized butterfly counts before the peeling process starts.
%     \item As stated later, we also utilize counting to reduce the \# wedges explored. Hence, important to study. 
%     \item Efficient algorithms exist, starting from xyz et al. who propose $\mathcal{O(something)}$ algorithm (shown in alg.\ref{}) for rectangle counting in unipartite graphs.
%     \item Wang et al. further propose sorting vertices for improved cache performance.
%     \item Can be parallelized over the starting vertices in wedges (line x in alg.\ref{}).
% \end{itemize}
\begin{algorithm}[htbp]
	\caption{Counting per-vertex butterflies (\texttt{pvBcnt})}
	\label{alg:counting}
	\begin{algorithmic}[1]
	    \Statex{\textbf{Input:} Bipartite Graph $G(W=(U, V), E)$} 
	    \Statex{\textbf{Output:} Butterfly counts $\bowtie_u \forall u\in W$}
%	    \Statex{$p\rightarrow$ \# parallel threads}
        \State{Relabel vertices in $G$ in descending order of degree}
        \ParForEach{$u\in U\cup V$}
        \State{Sort $N_u$ in ascending order of new labels}
        \State{$\bowtie_w\leftarrow 0$}
        \EndParForEach
        \ParForEach{$sp\in U\cup V$} 
            \State{Initialize $wdg\_arr$ array to all zeros}
            \State{$nze \leftarrow \{\phi\}$, $nzw \leftarrow \{\phi\}$}
            \ForEach{$mp\in N_{sp}$}
                \ForEach{$ep\in N_{mp}$}
                    \If{$(ep\geq mp)$ or $(ep\geq sp)$}{ break}\EndIf
                    \State{$wdg\_arr[ep] \leftarrow wdg\_arr[ep] + 1$}
                    \State{$nze \leftarrow nze \cup \{ep\}$, $\ nzw \leftarrow nzw \cup \{(mp, ep)\}$}
                \EndForEach
            \EndForEach
            \ForEach{$u\in nze$} \Comment{\textit{same side contribution}}
                \State{$bcnt \leftarrow {wdg\_arr[u] \choose 2}$}
                \State{Atomically add $bcnt$ to $\bowtie_u$ and $\bowtie_{sp}$}
            \EndForEach
            \ForEach{$(u,v)\in nzw$} \Comment{\textit{opp. side contribution}}
                \State{$bcnt \leftarrow wdg\_arr[v] -1$}
                \State{Atomically add $bcnt$ to $\bowtie_u$}
            \EndForEach
        \EndParForEach
	\end{algorithmic}
\end{algorithm}

\subsection{Tip Decomposition}\label{sec:bottomup}
Sariyuce et al.\cite{sariyucePeeling} introduced $k-$tips to identify vertex-induced subgraphs with large number of butterflies. They formally define it as follows:
\begin{definition}\label{def:ktip}
A bipartite subgraph $H = (U',V, E) \subseteq G$, induced on $U$, is a \textbf{k-tip} iff
\begin{itemize}[leftmargin=*]
    \itemsep0em
    \item each vertex $u \in U'$ participates in at least k butterflies,
    \item each vertex-pair $(u,u')\in U'$ is connected by a series of butterflies,
    \item H is maximal i.e. no other k-tip subsumes H.
\end{itemize}
\end{definition}

$k-$tips are hierarchical -- a $k-$tip overlaps with $k'-$tips for all $k'<=k$. Therefore, storing all $k-$tips is inefficient and often, infeasible. \textbf{Tip number} $\theta_u$ is defined as the maximum $k$ for which $u$ is present in a $k-$tip. Tip numbers provide a space-efficient representation of $k-$tip hierarchy with quick retrieval. In this paper, we study the problem of finding tip numbers, known as \textbf{tip decomposition}.

Algorithms in current practice use Bottom-Up Peeling (\texttt{BUP}) for tip-decomposition, as shown in alg.\ref{alg:bottomup}. It initializes vertex support using per-vertex butterfly counting, and then iteratively peels the vertices with minimum support until no vertex remains. When a vertex $u\in U$ is peeled, its support at that instant is recorded as its tip number $\theta_u$. Further, for every vertex $u'$ with $\bowtie_{u,u'}>0$ shared butterflies, the support of $u'$ is decreased by $\bowtie_{u,u'}$
% $\bowtie(u, v, u', v')\in\ \bowtie_u$, $\bowtie_{u'}$ is decremented by $1$ 
(capped at $\theta_u$). Thus, tip numbers are assigned in a non-decreasing order. The complexity of bottom-up peeling (alg.\ref{alg:bottomup}), dominated by wedge traversal (lines 8-10), is $\mathcal{O}\left(\sum_{u\in U}\sum_{v\in N_u}d_v\right)$.

\begin{algorithm}[]
	\caption{Tip decomposition using bottom-up peeling (\texttt{BUP})}
	\label{alg:bottomup}
	\begin{algorithmic}[1]
	    \Statex{\textbf{Input:} Bipartite graph $G(W=(U, V), E)$} 
	    \Statex{\textbf{Output:} Tip numbers $\theta_u\ \forall\ u\in U$}
%	    \Statex{$p\rightarrow$ \# parallel threads}
        \State{$\{\bowtie_U,\bowtie_V\}\leftarrow$ \texttt{pvBcnt($G$)}}\Comment{\textit{Initial count}}
        \While{$U\neq\ \{\phi \}$}\Comment{\textit{Peel}}
            \State{let $u\in U$ be the vertex with minimum support $\bowtie_u$}
            \State{$\theta_u \leftarrow\  \bowtie_u,\ \ U\leftarrow U\setminus \{u\}$}
            \State{\texttt{\textsc{update}(}$u, \theta_u, \bowtie_U, G$\texttt{)}}    
        \EndWhile
        \Function{\texttt{update}}{$u, \theta_u, \bowtie_U, G$}\label{func:update}
            \State{Initialize hashmap $wdg\_arr$ to all zeros}
            \ForEach{$v\in N_u$}
                \ForEach{$u' \in N_v \setminus \{u\}$}
                    \State{$wdg\_arr[u'] \leftarrow wdg\_arr[u'] + 1$}
                \EndForEach
            \EndForEach 
            \ForEach{$u'\in wdg\_arr$}\Comment{\textit{Update Support}}
                \State{$\bowtie_{u,u'} \leftarrow {wdg\_arr[u']\choose 2}$}\Comment{shared butterflies}
                \State{$\bowtie_{u'} \leftarrow \max\{\theta_u,\ \bowtie_{u'}-\bowtie_{u,u'}\}$}
            \EndForEach
        \EndFunction
	\end{algorithmic}
\end{algorithm}

\subsubsection{Challenges}\label{sec:challenges}
Tip decomposition is computationally very expensive and parallel computing is widely used to accelerate such workloads. However, the state-of-the-art parallel tip decomposition framework P\textsc{ar}B\textsc{utterfly}~\cite{shiParbutterfly, julienne} only utilizes parallelism within each peeling iteration by concurrently peeling all vertices with minimum support value.
%\kldelete{Instead of peeling one vertex at a time, the state-of-the-art parallel implementation P\textsc{ar}B\textsc{utterfly}~\cite{shiParbutterfly} peels all vertices with minimum support concurrently in a single iteration (line 2, alg.\ref{alg:bottomup}). A parallel heap or bucketing structure is leveraged to efficiently retrieve vertices with smallest support \cite{julienne, parallelHeap}.}
This restrictive approach is adopted due to the following sequential 
dependency between iterations -- \textit{support updates computed in an 
iteration guide the choice of vertices to peel in the subsequent iterations.} 
\kledit{As shown in \cite{shiParbutterfly}, P\textsc{ar}B\textsc{utterfly} is work-efficient with a complexity of 
$\mathcal{O}\left(\sum_{u\in U}\sum_{v\in N_u}{d_v} + \rho_v\log^2{m}\right)$,\\ where $\rho_v$ is the number of peeling iterations. However, its 
scalability is limited in practice due to the following drawbacks:}

\begin{enumerate}[leftmargin=*]
    \item Alg.~\ref{alg:bottomup} incurs large number of iterations and low workload per individual iteration. The resulting synchronization and load imbalance render simple parallelization ineffective for \looseness=-1acceleration.\\
    \textbf{Objective 1} -- Design an efficient parallel algorithm for tip decomposition that reduces thread synchronizations.
    %parallelizes workload across iterations.
    \item Alg.~\ref{alg:bottomup} explores all wedges with end-points in $U$. This is computationally expensive and can make it infeasible even for a scalable parallel algorithm, to execute in reasonable time.\\
    \textbf{Objective 2} -- Reduce the amount of wedge traversal. 
\end{enumerate}

%% file: receipt.tex
\section{REfine CoarsE-grained IndePendent Tasks Algorithm}\label{sec:receipt}

\begin{figure*}[htbp]
    \centering
\includegraphics[width=0.85\linewidth]{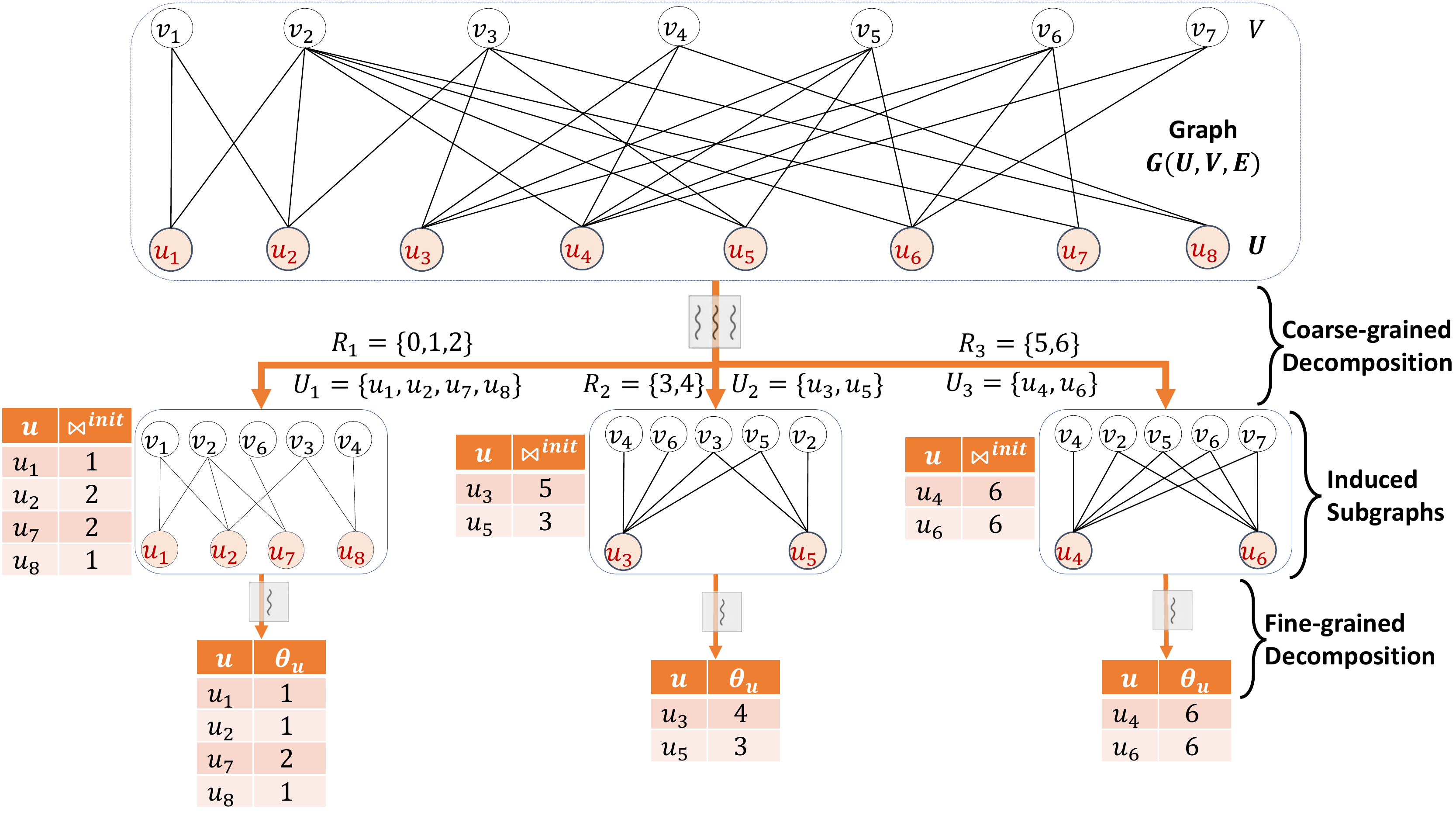}     
\caption{Graphical illustration of tip decomposition using RECEIPT. Coarse-grained Decomposition partitions $U$ into three vertex-subsets $U_1, U_2$ and $U_3$ whose tip numbers belonging to ranges $R_1, R_2$ and $R_3$, respectively. It processes each peeling iteration in parallel. Fine-grained decomposition creates subgraphs induced on $U_1, U_2$ and $U_3$, initializes vertex support using $\bowtie^{init}$ and peels each of them sequentially. It processes multiple subgraphs concurrently. Note that $G$ has $38$ wedges  whereas the induced subgraphs collectively have only $11$ wedges.}
    \label{fig:example}
\end{figure*}

%\kldelete{Targeting objective $1$ (sec.\ref{sec:challenges}), we first develop}
In this section, we present 
a novel shared-memory parallel algorithm for tip decomposition -- REfine CoarsE-grained IndePendent Tasks (RECEIPT), that drastically reduces the number of parallel peeling iterations (objective 1, sec.\ref{sec:challenges}). The fundamental insight  underlying RECEIPT is that $\theta_u$ only depends on the number of butterflies that $u$ shares with vertices having tip numbers \textit{no less than} $\theta_u$. Therefore, if we initialize the support $\bowtie_u$ to total butterflies of $u$ in $G$, only the following are required to compute $\theta_u$:
\begin{itemize}
    \item The \textit{aggregate} effect  of peeling all vertices $v$ with $\theta_v < \theta_u$ on $\bowtie_u$: $\bowtie_u$ would be $\geq \theta_u$ after such deletion. 
    \item The effect of deleting vertices $v$ with $\theta_v = \theta_u$ on $\bowtie_u$: $\bowtie_u$ would be $\leq \theta_u$ after such deletion.
\end{itemize}

This insight allows us to eliminate the major 
bottleneck for parallel tip decomposition i.e. the 
constraint of deleting only minimum support 
vertices in any peeling iteration. In order
to find vertices with tip number greater than or equal to $\theta$, all vertices
with tip numbers less than $\theta$ can be peeled simultaneously, providing sufficient parallelism. However, for 
every $\theta\in \{0,1\dots\theta^{max}\}$, 
peeling all the vertices $v$ with 
$\theta_v<\theta$ and computing corresponding support updates will make the algorithm 
extremely inefficient.
To avoid this inefficiency, 
RECEIPT follows a 2-step approach as shown in 
fig.\ref{fig:example}.

In the first step, it divides the 
spectrum\footnote{RECEIPT does not assume prior
knowledge of maximum tip number value and 
computes an upper bound during the execution.} 
of tip numbers $[0, \theta^{max}]$ into $P$ 
smaller ranges ${R_1, R_2\dots R_P}$, \kledit{where $P$ is a 
user-defined parameter}. A range 
$R_i$ is defined by the set of integers in 
$[\theta(i), \theta(i+1))$ with boundary conditions 
%\kldelete{$\theta(i+1)$ is \textit{strictly greater than}
%$\theta(i)$, with boundary conditions }
$\theta(1) = 0$ and $\theta(P+1) > \theta^{max}$. 
Note that the ranges have no overlap i.e. for 
any pair $(i,j)$, $i\neq j$ implies $R_i \cap R_j = \{\phi\}$.
Corresponding to each range $R_i$, RECEIPT CD also
finds the subset of vertices $U_i$ whose 
tip-numbers belong to that range i.e. $U_i = 
\cup_{u\in U}\{u |\ \theta_u \in R_i\}$. In other words, instead of
finding the exact tip number $\theta_u$ of a 
vertex $u$, the first step in RECEIPT computes 
\textit{bounds} on $\theta_u$, by finding its 
range affiliation using peeling. Therefore, this step is 
named \textbf{Coarse-grained Decomposition (RECEIPT CD)}. 
The absence of overlap between the ranges allows
each subset to be peeled independently of others for
exact tip number computation in a later step.
% \kldelete{For every subset $U_i$, RECEIPT CD also records the 
% support of vertices in $U_i$ after the 
% immediately preceding subset $U_{i-1}$ is 
% completely peeled. This 
% enables efficient initialization of support 
% values in a later step.}

RECEIPT CD has a stark difference from conventional bottom-up approach: instead of peeling vertices with minimum support, every iteration concurrently peels \textit{all vertices} with support value in a \textit{broad range}. \kledit{Setting $P\ll \theta^{max}$ ensures a large amount of vertices peeled per iteration (\textit{sufficient parallel workload}) and significantly less number of iterations (\textit{dramatically less synchronization}) compared to parallel variants of bottom-up peeling (\looseness=-1alg.\ref{alg:bottomup})}.

The next step finds the exact tip numbers of vertices and is termed \textbf{Fine-grained Decomposition (RECEIPT FD)}. \kledit{The key idea behind RECEIPT FD is as follows -- for each vertex $u$, if we know the number of butterflies it shares with vertices in its own subset $U_{i}$ and in subsets with higher tip number ranges ($\cup_{j>i}\{U_{j}\}$), then every subset can be peeled independently of others. RECEIPT FD exploits this independence to concurrently process multiple vertex subsets by simultaneously employing sequential bottom up peeling on the subgraphs induced by these subsets. 
Setting $P\gg T$ ensures that RECEIPT FD can be efficiently parallelized. Thus, RECEIPT avoids strict sequential dependencies across peeling iterations and systematically parallelizes tip decomposition.}
% \kldelete{It utilizes the fact that if the support of vertices in $U_i$ is correctly initialized, it can be peeled independently of other vertex subsets.
% Specifically, we need to ensure that initial support values of all $u\in U_i$ reflect the number of butterflies shared between $u$ and vertices with tip numbers greater than or equal to $\theta(i)$.}
% \kldelete{be the same as $\bowtie_u$ in alg.\ref{alg:bottomup} just before the first vertex in $U_i$ is peeled.
% After support initialization, RECEIPT FD concurrently processes multiple vertex subsets using $T$ threads. Each \kledit{thread peels one subset at a time using sequential bottom-up peeling to compute the exact tip numbers.}}
% \kldelete{of which works on one subset at a time. To find the exact tip numbers of vertices in $U_i$, a thread sequentially peels vertices in $U_i$ using the bottom-up peeling.} 

\kledit{Since each butterfly has two vertices in $U$, butterflies with vertices in different subsets are not preserved in the induced subgraphs. Hence, butterfly counting on a subgraph induced on $U_i$ will not account for butterflies shared between $u\in U_i$ and vertices in subsets with higher tip number ranges. However, we note that when $U_{i-1}$ is completely peeled in RECEIPT CD, the support of a vertex $u\in U$ reflects the number of butterflies it shares with remaining vertices i.e. $\cup_{j\geq i}\{U_{j}\}$. This is precisely the initial butterfly count for $u$ as required in RECEIPT FD. Hence, we store these values during RECEIPT CD (in $\bowtie^{init}$ vector as shown in fig.\ref{fig:example}) and use them for support initialization in RECEIPT FD.}

The two-step approach of RECEIPT can potentially double the workload as each wedge may be traversed once in both steps.
%However, we utilize the \textit{independence} of vertex subsets to drastically reduce the computational overhead of FD. 
However, we note that since each subset is processed independently of others, support updates are \textit{not communicated between the subsets}. Hence, when peeling a subset $U_i$, we only need to traverse
the wedges with both endpoints in $U_i$. Therefore, RECEIPT FD
%Therefore, while processing $U_i$, it is \textit{unnecessary} to traverse the wedges shared with vertices in other subsets. To prevent such wasteful computation, we 
first creates an \textit{induced subgraph} on $U_i$ and only explores the wedges in that subgraph. This dramatically reduces the amount of work done in RECEIPT FD. For example, in fig.\ref{fig:example}, the original graph $G$ has $38$ wedges with endpoints in $U$, whereas the three subgraphs collectively have only $11$ such wedges which will be traversed during RECEIPT FD. 

%executes sequential bottom-up peeling on the subgraph induced on $U_i$. 

% \begin{itemize}
%     \item Explain the basics of the algorithm
%     \item Figure for explaining what happens
% \end{itemize}
\subsection{Coarse-grained Decomposition}\label{sec:CD}
Alg.\ref{alg:cd} depicts the pseudocode for Coarse-grained Decomposition (RECEIPT CD). 
%(\texttt{RECEIPT\_CD}), 
% It takes a bipartite graph $G(U,V,E)$ as input and partitions $U$ and the tip number of vertices in $U$ into $P$ subsets. 
It takes a bipartite graph $G(U,V,E)$ as input and partitions the vertex set $U$ into $P$ subsets with different tip number ranges.
Prior to creating vertex subsets, RECEIPT CD uses per-vertex counting (alg.\ref{alg:counting}) to initialize the support of vertices in $U$.
%It begins with initializing all the vertex subsets to empty sets. Support value of vertices are (butterfly counts) initialized using butterfly counting (alg.\ref{alg:counting}). 
%Thereafter, it computes the tip number ranges and peels corresponding vertex subsets.
%\kledit{Note that the peeling procedure only explores wedges with endpoints in the set $U$ for which tip decomposition is being computed.}

The first step in computing a subset $U_i$ is to find the tip number range $R_i=[\theta(i), \theta(i+1))$. 
\kledit{Ideally, the ranges should be computed such that the number of wedges in the induced subgraphs are uniform for all $P$ subsets (to ensure load balance 
during RECEIPT FD). However, induced subgraphs are not known prior to
actual partitioning. Secondly, exact tip numbers are not known either and hence, vertices in $U_i$ cannot be determined prior to partitioning, for different values of $\theta(i+1)$. 
Considering these challenges, RECEIPT CD uses two proxies for range determination (lines 16-21): 
\begin{enumerate*}
    \item the number of wedges in the original graph as a proxy for the wedges in induced subgraphs, and
    \item current support of vertices as a proxy for tip numbers.
\end{enumerate*}
% Hence, we use the number of wedges in the original graph as a proxy for the purpose of range computation. 
Now, to balance the wedge counts across subsets, RECEIPT CD 
aggregates the \textit{wedge count (in $G$)} of vertices in bins 
corresponding to their \textit{current support} and computes a prefix
sum over the bins (Ref: proof of theorem~\ref{theorem:complexity}). For any unique support value $\theta$,
the prefix output now represents the total wedge count of vertices with support $\leq \theta$ (line 19).
The upper bound $\theta(i+1)$ is then chosen such that the prefix 
output for $\theta(i+1)$ is 
close to (but no less than) the average wedges 
per subset (line 20).
% Therefore the upper bound $\theta(i+1)$ is computed such that \textit{total wedge count (in $G$)} of vertices with current support in range $R_i$, is 
% close to (but no less than) the average wedges 
% per subset (lines 4, 16-21). This is achieved by aggregating total wedge count of vertices for every unique support value and then computing a prefix sum over these counts.
%Consequently, 
%every subset will cover at least average number
%of wedges and RECEIPT CD % \texttt{RECEIPT\_CD} 
%will peel the entire set $U$ at the end of 
%$P$ subsets.
} 
% \kldelete{using the \texttt{\textsc{findHi}} function (line 8). Since the parallelism in FD stage will be across the vertex subsets, we target imbibing equal number of incident wedges $tgt$~(line 4) for all subsets. 
% %Note that the complexity of tip decomposition is dominated by wedge traversal. Hence, we identify the workload as the number of wedges incident on $U_i$. 
% The \texttt{\textsc{findHi}} function computes the smallest support value $\theta$, such that total wedges incident on all vertices with support \textit{no greater than} $\theta$, is just larger than the  target wedges (line 20). Note that every subset covers at least average number of wedges and we are guaranteed to peel the entire set $U$ at the end of $P$ subsets.}

After finding the range, RECEIPT CD
%\texttt{RECEIPT\_CD} 
iteratively peels the vertices and adds them to subset $U_i$ (lines 9-14).
\kledit{In each iteration, it peels $activeSet$ -- the set of all vertices  with 
support in the entire range $R_i$.}
\kledit{This is unlike bottom-up peeling where vertices with a single (minimum) support value are peeled in an iteration.}
% \kldelete{This is similar to bottom-up peeling approach but instead of removing vertices 
% with a single support value (line 3, alg.\ref{alg:bottomup}), we remove a large number of 
% vertices with support values in a broad range (line 10, alg.\ref{alg:cd}).} 
Thus, RECEIPT CD
%\texttt{RECEIPT\_CD}
enjoys higher workload per iteration which enables efficient parallelization. 
% \kldelete{Note that the parallel bucketing structure of Julienne is not work-efficient in 
% the context of tip decomposition\cite{julienne, shiParbutterfly} and hence, not used in 
% \texttt{RECEIPT\_CD}.
% % Therefore, we do not use bucketing to retrieve $activeSet$ (line 10) in peeling iterations. 
% Instead, we scan all vertices in $U$ to intialize the $activeSet$ in first peeling iteration of every $U_i$ (line 9).}
% %Every iteration peels a set of vertices -- $activeSet$ and updates the support of remaining vertices in $U$. 
For the first peeling iteration of every subset $U_i$, RECEIPT CD scans all vertices in $U$ to initialize $activeSet$ (line 9).
In subsequent iterations, $activeSet$ is constructed by tracking only those vertices whose support has been updated. \kledit{For correctness of parallel algorithm, the \texttt{\textsc{update}} routine used in RECEIPT CD (line 13) uses atomic operations to decrease vertex \looseness=-1supports.}

Apart from tip number ranges and partitions of $U$, 
RECEIPT CD %\texttt{RECEIPT\_CD} 
also outputs an array $\bowtie^{init}_U$, which is used to 
initialize support in RECEIPT FD (sec.\ref{sec:receipt}). Before peeling for a
subset begins, it copies the current support of 
remaining vertices into $\bowtie^{init}_U$ (line 7). 
Thus, for a vertex $u\in U_i$, $\bowtie^{init}_u$ 
indicates the support of $u$ 
\begin{enumerate*}[label=(\alph*)]
    \item after \textit{all} vertices in $U_{i-1}$ are peeled, and
    \item before \textit{any} vertex in $U_{i}$ is \looseness=-1peeled.
\end{enumerate*}

\begin{algorithm}[htbp]
	\caption{Coarse-grained Decomposition (RECEIPT CD)}
	\label{alg:cd}
	\begin{algorithmic}[1]
	    \Statex{\textbf{Input:} Bipartite graph $G(U, V, E)$, \# partitions $P$} 
	    \Statex{\textbf{Output:} Ranges $\{\theta(1), \theta(2)\dots \theta(P+1)\}$, Vertex Subsets $\{U_1, U_2\dots U_P\}$, Support initialization vector $\bowtie^{init}_U$}
%	    \Statex{$p\rightarrow$ \# parallel threads}
        \Statex{$w[u] \leftarrow$ number of wedges in $G$ with endpoint $u$}
        \State{$U_i \leftarrow \{\phi\}\ \forall\ i\in \{1,2\dots P\}$}
        \State{$\{\bowtie_U,\bowtie_V\}\leftarrow$ \texttt{pvBcnt($G$)}}
        \State{$\theta(1)\leftarrow0$,  $\ i\leftarrow 1$}
        \State{$tgt\leftarrow\frac{\sum_{u\in U}w[u]}{P}$}\Comment{\textit{average wedges per subset}}
        \While{$U\neq\ \{\phi \}$ \textbf{and} $i\leq P$}
            \ParForEach{$u\in U$}\Comment{\textit{Support Init}}
                \State{$\bowtie^{init}_u \leftarrow\ \bowtie_u$}
            \EndParForEach
            \State{$\theta(i+1)\leftarrow\text{\texttt{\textsc{findHi}(}}U, \bowtie_U, w, tgt\text{\texttt{)}}$}\Comment{\textit{Upper Bound}}
            %\State{$R_p \leftarrow \{\theta^p, \theta^p+1\dots \theta^{p+1}-1\}$}
            %\State{$activeSet \leftarrow \cup_{u\in U}\{u\ |\  \theta(i)\leq\ \bowtie_u\ <\theta(i+1)\}$} 
            \State{$activeSet\leftarrow$vertices with support in $\big[\theta(i),\theta(i+1)\big)$}
            
            %\Comment{\textit{$1^{st}$ iteration}}
            
            \While{$activeSet\neq \{\phi \}$}\Comment{\textit{Peel Range}}
                %\State{$activeSet \leftarrow$ vertices $u\in U$ s.t. $\theta(i)\leq\ \bowtie_u\ <\theta(i+1)$}
                % \If{$activeSet =\{\phi\}$}{ break}\EndIf
                \State{$U_i \leftarrow U_i\cup activeSet$, $\ U \leftarrow U\setminus activeSet$}
                \ParForEach{$u\in activeSet$}
                    \State{\texttt{\textsc{update(}$u, \theta(i), \bowtie_U, G$\texttt{)}}} \Comment{Ref: alg.\ref{alg:bottomup}}
                \EndParForEach                \State{$activeSet\leftarrow$vertices with support in $\big[\theta(i),\theta(i+1)\big)$}
                % \State{$activeSet \leftarrow \cup_{u\in U}\{u\ |\  \theta(i)\leq\ \bowtie_u\ <\theta(i+1)\}$}
            \EndWhile
            \State{$i\leftarrow i+1$}
        \EndWhile
        
        \Function{\texttt{findHi}}{$U, \bowtie_U, w, tgt $}\label{func:findhi}
            \State{Initialize hashmap $work$ to all zeros}
            \ForEach{$\bowtie\ \in\  \bowtie_U$} %\Comment{histogramming}
                %\State{$work[\bowtie_u] \leftarrow work[\bowtie_u] + \sum_{u\in U}\left(w[u]\right)$}
                \State{$work[\bowtie] \leftarrow \sum_{u\in U}\left(w[u]\cdot \mathbbm{1}(\bowtie_u \leq\ \bowtie)\right)$}
            \EndForEach
            \State{$\theta\leftarrow$ $\arg\min\left(\bowtie\right)$ such that $work[\bowtie] \geq tgt$}
            %\kldelete{\State{$\theta\leftarrow$ smallest $\bowtie\ \in\  \bowtie_U$ s.t. $work[\bowtie] \geq tgt$}}
            \State{return $\theta+1$}
            % \ForEach{$\bowtie \in work$}\Comment{Prefix Sum}
            %     \State{$work[\bowtie]\leftarrow \sum_{\bowtie_u<\bowtie}work[\bowtie_u]$}
            % \EndForEach
        \EndFunction
	\end{algorithmic}
\end{algorithm}

\subsubsection{Adaptive Range Determination}\label{sec:adaptive}
% \kldelete{In alg.\ref{alg:cd}, the \texttt{\textsc{findHi}} function computes an upper bound based on the 
% current support of the vertices. For example the upper bound $\theta(i+1)$ for subset $U_i$ is computed in such 
% a way that $tgt + \Delta$ wedges are incident on vertices with  support in range $[\theta(i), \theta(i+1))$. In 
% the first peeling iteration (lines 9-11), all of these vertices will be deleted. As the support of other 
% vertices decreases due to peeling, more vertices can get added to $U_i$ (line 14).}
\kledit{Since range determination uses current support of vertices as a proxy for tip numbers, $tgt$-- the target wedge count for a subset $U_i$, is 
covered by the vertices added to $U_i$ in the very first
peeling iteration. Hence, ${\sum_{u \in U_i} w[u] \geq tgt}$, where $w[u]$ is the wedge count of $u$ in $G$. Note that as the support of other vertices 
decreases after this iteration, more vertices may get 
added to $U_i$ and total wedge count of vertices in the final subset $U_i$ can be significantly higher than $tgt$. This could result in some 
subsets having very high workload.  Moreover, it is possible that $U$ gets completely deleted in $\ll P$ 
subsets, thus restricting the parallelism available during RECEIPT FD}. To avoid this scenario, we implement 
two-way adaptive range determination:
\begin{enumerate}[leftmargin=*]
    \itemsep0em
    \item \kledit{Instead of statically computing an averge target $tgt$, we 
    dynamically update $tgt$ for every subset based on the wedge count of remaining vertices in $U$ and the remaining number of subsets to create.} {\it If some subsets cover a large number of wedges, the 
    target for future subsets is automatically reduced, thereby preventing a situation where all 
    vertices get peeled in much less than $P$ subsets.}
    \item \kledit{A subset $U_i$ can cover significantly larger number of wedges than the target $tgt$. RECEIPT CD assumes predictive local behavior i.e. subset $U_{i+1}$ will exhibit similar behavior to $U_i$. Therefore, to balance the wedge counts of subsets, RECEIPT dynamically scales the target wedge count for $U_{i+1}$ with a scaling factor $s_i=\frac{tgt}{\sum_{u\in U_i}w[u]}\leq 1$. Note that $s_i$ quantifies the overshooting of target wedges in $U_i$.}
    % \item For every subset $U_i$, we compute a scaling factor $s_i$ to quantify the overshooting of target wedges as follows: $s_i=\frac{tgt}{\sum_{u\in U_i}w[u]}\leq 1$. \kledit{$s_i$ is used to scale the target wedges for the next subset $U_{i+1}$.
    % If $U_i$ overshoots its target by a large amount, target wedges for $U_{i+1}$ will be reduced \kledit{in an expectation that $U_{i+1}$ will also overshoot the specified target}.
    % This enables RECEIPT CD to account for additional vertices inserted in the subsets after their first peeling iteration.}
    %While computing upper bound for $U_{i+1}$, we scale the target obtained from step~1 with the factor $s_i$. 
    % \kldelete{In case the tip numbers are concentrated around a small intermediate range, it prevents few partitions from getting too much workload.}
    %\kledit{Note that scaling factor of $U_i$ only affects the target for subset $U_{i+1}$.}
\end{enumerate}

After $P$ partitions, if some vertices still remain in $U$, RECEIPT CD puts all of them in a single subset $U_{P+1}$ and increments $P$.
%to indicate that an additional subset was \looseness=-1created.
% \begin{itemize}
%     \item Explain the procedure in detail - where is the parallelism etc.
%     \item Range decided to have equal distribution of wedges in subsets
%     \item Algorithm
% \end{itemize}

\subsection{Fine-grained Decomposition}\label{sec:fd}
 Alg.\ref{alg:fd} presents the pseudocode for Fine-grained Decomposition (RECEIPT FD),
 %(\texttt{RECEIPT\_FD}). 
 which takes as input the vertex subsets and the tip number ranges created by RECEIPT CD, 
 and computes the exact tip numbers. 
%\texttt{RECEIPT\_FD} 
%\kldelete{It starts with creating a task queue of subset IDs (line 1) and a parallel environment by launching all the threads in the system. The threads are then allocated subsets from this queue for exact tip number computation.}
It creates a task queue of subset IDs from which threads are exclusively allocated vertex subsets to peel (line 4). Before peeling $U_i$, a thread initializes the support of vertices in $U_i$ from the $\bowtie^{init}_U$ vector and induces a subgraph $G_i$ on $W_i=(U_i, V)$ (lines 5-6). Thereafter, sequential bottom-up peeling is applied on $G_i$ for tip decomposition of $U_i$ (lines 7-10).
%For efficient retrieval of minimum support vertices in bottom-up peeling, a fibonacci heap can be used.

\begin{algorithm}[htbp]
	\caption{Fine-grained Decomposition (RECEIPT FD)}
	\label{alg:fd}
	\begin{algorithmic}[1]
	    \Statex{\textbf{Input:} Bipartite graph $G(U, V, E)$, \# partitions $P$, Vertex Subsets $\{U_1, U_2\dots U_P\}$, Support initialization vector $\bowtie^{init}_U$, \# threads $T$}
	    \Statex{\textbf{Output:} Tip number $\theta_u$ for each $u\in U$} 
        %\Statex{Let $\wedge_i \leftarrow$ wedges incident on vertices in $U_i$}
        \State{Insert the integers $i\in \{1,2\dots P\}$ in a queue $Q$
        %$Q\leftarrow$ integers $i\in \{1,2\dots P\}$ in decreasing order of $\wedge_i$
        }
        \ParFor{$thread\_id = 1,2\dots T$}
            \While{$Q$ is not empty}
                \State{Atomically pop $i$ from $Q$}
                %\State{Create $G_i(U_i, V, E_i= \cup_{(u,v)\in E}\{(u,v)| u\in U_i\})$}\Comment{\textit{Induce Subgraph}}
                \State{$G_i \leftarrow$ subgraph induced by $W_i = (U_i, V)$}
                \State{$\bowtie_{U_i}\  \leftarrow\  \bowtie^{init}_{U_i}$}       \Comment{\textit{Initialize Support}}
                \While{$U_i \neq \{\phi\}$}\Comment{\textit{Peel}}
                    \State{$u \leftarrow \arg\min_{u\in U_i}\{\bowtie_u \}$}    \State{$\theta_u\leftarrow\ \bowtie_u$, $\ U_i\leftarrow U_i\setminus \{u\}$}
                    \State{\texttt{\textsc{update(}}$u,\theta_u,\bowtie_{U_i}, G_i$\texttt{)}}
                \EndWhile
            \EndWhile
        \EndParFor
	\end{algorithmic}
\end{algorithm}
\subsubsection{Parallelization Strategy}\label{sec:schedule}
While adaptive range determination(sec.\ref{sec:adaptive}) tries to create subsets with
uniform wedges in $G$, the actual work per subset in FD depends on the wedges
in induced subgraphs $G_i(U_i, V, E_i)$ that can be non-uniform. \kledit{Therefore, to \textit{improve load balance} across threads, we use parallelization strategies inspired from Longest Processing Time scheduling rule which is a known $\frac{4}{3}$-approximation algorithm \cite{graham1969bounds}. However, exact processing time for peeling an induced subgraph $G_i$ is unknown. Instead, we use the number of wedges with endpoints in $U_i$ as a proxy along with runtime task scheduling as given below:}
% For \textit{efficient work distribution} across threads, we adopt the following strategies (analogous to Longest Processing Time scheduling rule which is a $\frac{4}{3}$-approximation algorithm)
%which is a $\frac{4}{3}$-approximation algorithm~\cite{graham1969bounds}
:
\begin{itemize}[leftmargin=*]
    \item \textit{Dynamic task allocation$\rightarrow$} Threads atomically pop unique subset IDs from the  task queue when they become idle during runtime (line 4). Thus, all threads are busy until every subset is scheduled.
    %threads getting low workload subsets can contribute effectively by processing more number of subsets.
    \item \textit{Workload-aware Scheduling$\rightarrow$} We sort the subset IDs in task queue in decreasing order of their wedge counts. Thus, the subsets with highest workload 
    (wedges) get scheduled first and the threads processing them naturally receive fewer
    tasks in the future.
    Fig.\ref{fig:schedule} shows how workload-aware scheduling can tremendously improve the efficiency of dynamic allocation.
\end{itemize}
% We implement a \textit{dynamic scheduling policy} for assigning tasks to the threads in a. When a thread $t_{id}$ becomes idle, it fetches the ID $i$ of an unprocessed vertex subset from the queue (line 4) for peeling. This provides significantly improved load balancing compared to static \looseness=-1scheduling.
 
%  The two-step approach in RECEIPT can potentially double the workload as all wedges are traversed once each in \texttt{RECEIPT\_CD} and \texttt{RECEIPT\_FD}. However, we 
%  Before peeling $U_i$, $t_{id}$ creates an \textit{induced subgraph} $G_i$ consisting of edges incident on vertices in $U_i$. Since each subset is peeled independently of others, we do not propagate updates to vertices in
\begin{figure}[htbp]
    \centering
\includegraphics[width=0.9\linewidth]{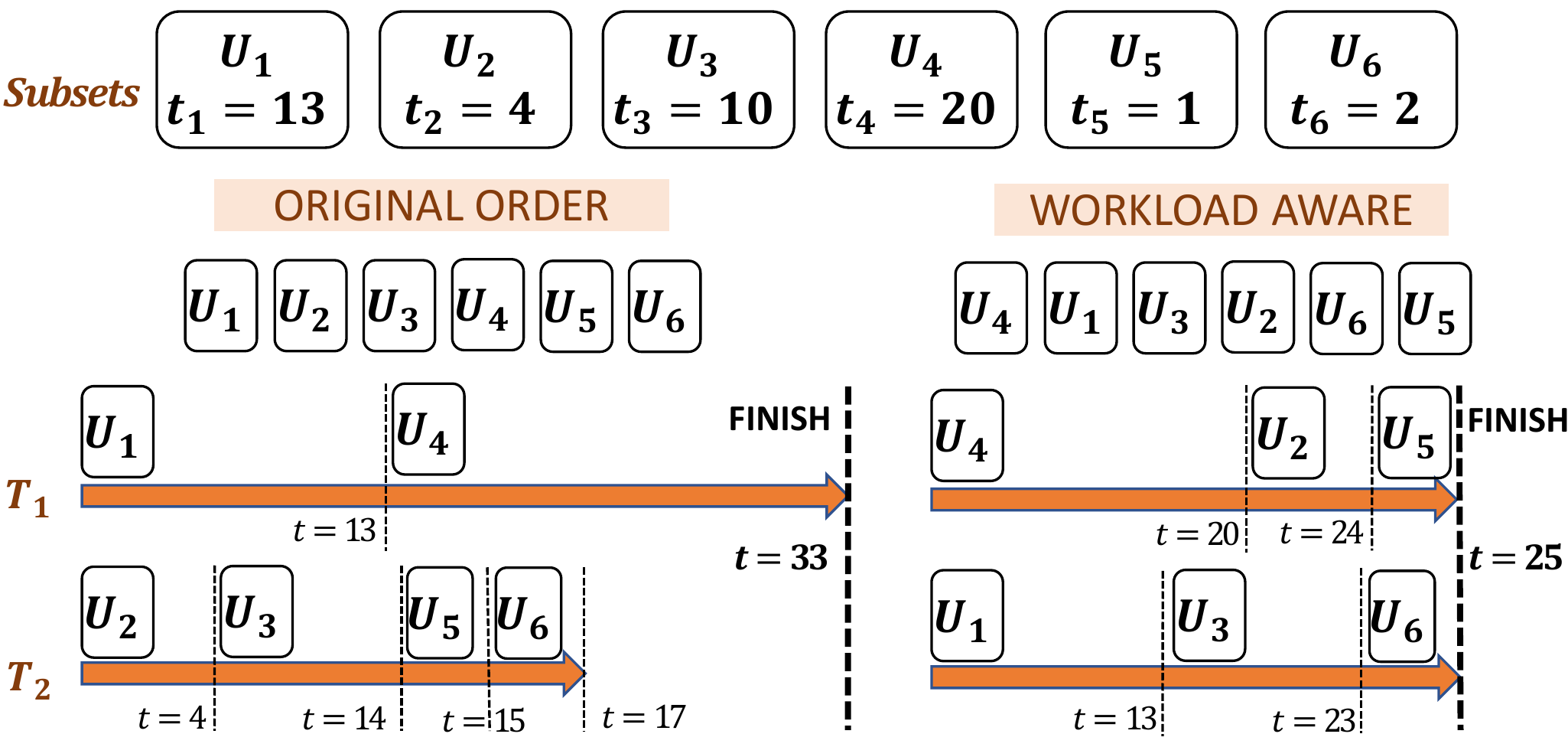}     
\caption{Benefits of Workload-aware Scheduling (WaS) in a $2$-thread ($T_1$ and $T_2$) system. Top row shows vertex subsets with time required to peel them. Dynamic allocation without WaS finishes in $33$ units of time compared to $25$ units with WaS.}
    \label{fig:schedule}
\end{figure}

% \begin{itemize}
%     \item Explain the procedure in detail (algorithm)
%     \item To fetch the vertices with minimum tip-number, we use the fibonacci heap (O(1) updates and O(log(V)) delete). Bucketing of julienne is not work-efficient in this context \cite{} and bucketing procedure used by sariyuce pinar needs to create and probe several empty buckets.
%     \item Initial tip-numbers when vertices in a subset start peeling determined from previous step.
%     \subsubsection{Adaptive range determination}
%     \begin{itemize}
%         \item Problem - imbalanced tasks at the end of coarse-grained decomposition.
%         \item First adaptation - dynamically update the desired number of wedges per partition based on
%         the set of vertices deleted so far
%         Second adaptation - target less than what is desired based on how much you exceeded in the previous round.
%     \end{itemize}
% \end{itemize}

\subsection{Analysis}\label{sec:analysis}
In this section, we will prove the correctness of tip numbers computed by RECEIPT. We will also analyze its computational complexity and show that RECEIPT is work-efficient. We will exploit comparisons with sequential \texttt{BUP} (alg.\ref{alg:bottomup}) and hence, first establish the follwing lemma:

\begin{lemma}\label{lemma:independence}
In \texttt{BUP}, the support $\bowtie_u$ of a vertex $u$ at any time $t$ before the first vertex with tip number $\theta_u$ is peeled, depends on the cumulative effect of all vertices peeled till $t$ and is independent of the order in which they are peeled.
\end{lemma}
\begin{proof}
Let $\bowtie_u^G$ be the initial support of $u$ (butterflies of $u$ in $G$), $S$ be the set of 
vertices peeled till $t$ and $u'\in S$ be the most recently peeled vertex. Since \texttt{BUP} assigns tip 
numbers in a non-decreasing order, no vertex with tip number $\geq \theta_u$ would be peeled till $t$. 
Hence, $\theta_{u'} < \theta_u \leq \bowtie_u$. Since $\theta_{u'}$ was the minimum vertex support in the latest peeling iteration, $\bowtie_u =\ \bowtie_u^G + \sum_{v\in S}{\left(-\bowtie_{v,u}\right)}$. By commutativity of addition, this term is independent of the order in which $-\bowtie_{v,u}$ are added i.e. the order in which vertices in $S$ are peeled.
% \kldelete{By 
% the definition of tip number (sec.\ref{alg:bottomup}), $u$ currently participates in no less than 
% $\theta_u$ butterflies i.e. $\bowtie_u \geq \theta_u >\theta_{u'}$. Since 
% Therefore, the $\max$ capping on $\bowtie_u$ (line 13, alg.\ref{alg:bottomup}) has not been applied and $\bowtie_u =\  \bowtie_u^G + \sum_{u'\in U'}{\left(-\bowtie_{u',u}\right)}$. 
% By commutativity of addition, this term is independent of the order in which $-\bowtie_{u,u'}$ are added i.e. the order in which vertices in $U'$ are peeled.}
\end{proof}

Lemma~\ref{lemma:independence} highlights that whether \texttt{BUP} peels a set of vertices $S\subseteq U$ in its original order or in the same order as RECEIPT CD, the support of vertices with tip numbers higher than that of $S$ would be the same. Next we show that parallel processing does not affect the correctness of support updates. 

\begin{lemma}\label{lemma:parallel}
Given a set $S$ of vertices to be peeled in an iteration, the parallel peeling in RECEIPT CD (line 12-13, alg.\ref{alg:cd}) correctly updates the support of vertices.
%\kldelete{computes the support updates}.
\end{lemma}
\begin{proof}
%  \kldelete{Consider an arbitrary vertex $u$. 
%  and let $\bowtie_u(t)$ denote the support of $u$ after peeling iteration $j$.} 
 \kledit{Let $S$ be peeled in $j^{th}$ iteration which is a part of peeling iterations for subset $U_i$. Parallel updates are correct if for any vertex $u$ not yet assigned to any subset, support of $u$ decreases by exactly $\sum_{u' \in S}{\bowtie_{u', u}}$ in the $j^{th}$ iteration.
% \kldelete{The parallel updates are correct if $\bowtie_u(t) =\  \bowtie_u(t-1) - \delta$, where $\delta$ is the number of butterflies deleted in $t^{th}$ iteration, that contain $u$.}
 Since at most two vertices of a butterfly can be in $S$ (two vertices of a butterfly are in $V$), for any vertex pair $\left(u_1, u_2\right)$, either of the following is true in $j^{th}$ iteration:} 
 %Either of the following is true for the pair:
\begin{itemize}[leftmargin=*]
    \item \textit{No vertex in the pair is peeled} -- no updates are propagated from $u_1$ to $u_2$ and vice-versa.
    \item \textit{Exactly one vertex in the pair is peeled} -- \kledit{without loss of generality}, let $u_1\in S$. Peeling $u_1$ deletes exactly $\bowtie_{u_1, u_2}$ \textit{unique} butterflies incident on $u_2$ because they share $\bowtie_{u_1, u_2}$ butterflies and no vertex in $U\setminus \{u_1, u_2\}$ (and hence in $S$) participates in these butterflies.
    %$u_1$ and $u_2$ share $\bowtie_{u_1, u_2}$ butterflies. No other vertex in $U\setminus \{u_1, u_2\}$ (and hence in $S$) participates in these butterflies. Thus, peeling $u_1$ deletes $\bowtie_{u_1, u_2}$ \textit{unique} butterflies that are incident on $u_2$. 
    The \texttt{update} routine called for $u_1$ (line 13, alg.\ref{alg:cd}) also decreases support of $u_2$ by exactly $\bowtie_{u_1, u_2}$, until $\bowtie_{u_2} > \theta(i)$. \kledit{Since atomics are used to apply support updates, concurrent updates to same vertex do not conflict (sec.\ref{sec:CD}).} Thus, \texttt{update} routine calls for all vertices in $S$ cumulatively decrease $\bowtie_{u_2}$ by exactly $\sum_{u'\in S}{\bowtie_{u', u_2}}$, as long 
    as $\bowtie_{u_2} > \theta(i)$. If $\bowtie_{u_2} \leq \theta(i) < \theta(i+1)$, vertex subset for $u_2$ is already determined and any updates to $\bowtie_{u_2}$ have no impact. 
    \item \textit{Both $u_1$ and $u_2$ are peeled} -- vertex subset for $u_1$ and $u_2$ is determined. Any update to $\bowtie_{u_1}$ or $\bowtie_{u_2}$ has no effect.
\end{itemize}
% \kldelete{Since atomics are used to decrease support values, updates to a given vertex from concurrent  \texttt{\textsc{update}} function calls can be applied correctly without conflict.}
\end{proof}

Now, we will
%will use the correctness of \texttt{BUP} and the equivalence of transient states of RECEIPT CD and \texttt{BUP} to 
show that RECEIPT CD correctly computes the tip number range for every vertex. Finally, we will show that RECEIPT accurately computes the exact tip numbers for all vertices. \kledit{For clarity of explanation, we use $\bowtie_u(j)$ to denote the support of a vertex $u$ after $j^{th}$ peeling iteration in RECEIPT CD.}
%Next, we prove the correctness of vertex partitioning in RECEIPT CD (alg.\ref{alg:cd}).

% \begin{lemma}\label{lemma:cd1}
% There cannot exist a vertex $u$ such that $u\in U_i$ and $\theta_u \geq \theta(i+1)$.
% \end{lemma}
% \begin{proof}
% Let $j$ be the first iteration in RECEIPT CD that wrongly peels a set of vertices $S_{w}$ with tip numbers higher than the current range. Let $j$ be a part of peeling iterations for subset $U_i$ and $S_{r}$ be the set of vertices peeled up to iteration $j-1$. \kledit{Consider a vertex $u\in S_{w}$. Since $u$ is peeled in iteration $j$, $\bowtie_u(j-1) < \theta(i+1)$. By definition of $S_w$, $\theta_u \geq \theta(i+1)$. }
% \kldelete{Consider a vertex $u\in S_{w}$ and let $\bowtie_u(t)$ denote the support of $u$ at the end of $t^{th}$ iteration in RECEIPT CD. By definition of $S_w$, $\theta_u \geq \theta(i+1)$. Since $u$ is peeled in iteration $j$, $\bowtie_u(j-1) < \theta(i+1)$.}

% Since all vertices till $j-1$ iterations have been correctly peeled, $\theta_{u'} < \theta(i+1)\ \forall\ u'\in S_r$. 

% % Therefore, $S_r$ will be peeled before $u$ in \texttt{BUP} (alg.\ref{alg:bottomup}). From lemmas~\ref{lemma:independence} and \ref{lemma:parallel}, $\bowtie_u(j-1)$ is equal to the support of $u$ in \texttt{BUP} after $S_r$ is peeled. However, since ${\bowtie_u(j-1)}<\theta(i+1)$, tip number of $u$ as correctly computed by \texttt{BUP} will be less than $\theta(i+1)$ which is a contradiction. Thus, no such $u$ exists, $S_{w}=\{\phi\}$ and all vertices in $U_i$ have tip numbers less than $\theta(i+1)$.
% \end{proof}

\begin{lemma}\label{lemma:cd1}
There cannot exist a vertex $u$ such that $u\in U_i$ and $\theta_u \geq \theta(i+1)$.
\end{lemma}
\begin{proof}
Let $j$ be the first iteration in RECEIPT CD that wrongly peels a set of vertices $S_{w}$, \kledit{and assigns them to subset $U_i$ even though $\theta_u \geq \theta(i+1)\ \forall\ u\in S_w$. Let $S_{hi} \supseteq S_{w}$ be the set of all vertices with tip numbers $\geq \theta(i+1)$ and $S_{r}$ be the set of vertices peeled up to iteration $j-1$.  
% \kldelete{Consider a vertex $u\in S_{w}$ and let $\bowtie_u(t)$ denote the support of $u$ at the end of $t^{th}$ iteration in RECEIPT CD. By definition of $S_w$, $\theta_u \geq \theta(i+1)$. Since $u$ is peeled in iteration $j$, $\bowtie_u(j-1) < \theta(i+1)$.}
Since all vertices till $j-1$ iterations have been correctly peeled, $\theta_{u} < \theta(i+1)\ \forall\ u\in S_r$.
Hence, $S_{hi} \subseteq U\setminus S_r$. 

Consider a vertex $u\in S_{w}$. Since $u$ is peeled in iteration $j$, ${\bowtie_u(j-1)} < \theta(i+1)$. From lemma \ref{lemma:parallel}, $\bowtie_u(j-1)$ correctly represents the number of butterflies shared between $u$ and $U\setminus S_r$ (vertices remaining after $j-1$ iterations). Since $S_{hi} \subseteq U\setminus S_r$, $u$ participates in at most $\bowtie_u(j-1)$ butterflies with vertices in $S_{hi}$. By definition of tip-number (sec.\ref{sec:bottomup}),  $\theta_u \leq \bowtie_u(j-1) < \theta(i+1)$, which is a contradiction. Thus, no such $u$ exists, $S_{w}=\{\phi\}$ and all vertices in $U_i$ have tip numbers less than $\theta(i+1)$.}
\end{proof}

\begin{lemma}\label{lemma:cd2}
There cannot exist a vertex $u$ such that $u\in U_i$ and $\theta_u < \theta(i)$.
\end{lemma}
\begin{proof}
Let $i$ be the smallest integer for which there exists a set $S_w\neq \{\phi \}$ such that $\theta(i)\leq \theta_u < \theta(i+1)\ \forall\ u\in S_w$, but $u\in U_p$, where $p>i$. Let $j$ be the last iteration that peels vertices in $U_i$. 
%Let $\bowtie_u(t)$ denote the support of vertex $u$ at the end of $t^{th}$ peeling iteration in RECEIPT CD. 
Clearly, $\bowtie_u(j) \geq \theta(i+1)\ \forall\ u\in S_w$ otherwise $u$ would be peeled in or before iteration $j$, and will not be added to $U_p$. 

From lemma \ref{lemma:parallel}, $\bowtie_u(j)$ correctly represents the butterfly count of $u$  after vertices in $U_1\cup U_2\dots \cup U_i$ are deleted. In other words, every vertex in $S_w$ participates in at least $\theta(i+1)$ butterflies with vertices in $U_{i+1}\cup U_{i+2} \dots \cup U_{P}$ and hence, is a part of $\theta(i+1)$-tip (def.\ref{def:ktip}). Therefore, by the definition of tip number, $\theta_u \geq \theta(i+1)\ \forall\ u\in S_w$ which is a contradiction. 
\end{proof}

\begin{theorem}\label{theorem:cd}
RECEIPT CD (alg.\ref{alg:cd}) correctly computes the vertex-subsets corresponding to every tip number range.
\end{theorem}
\begin{proof}
Follows directly from lemmas~\ref{lemma:cd1} and \ref{lemma:cd2}.
\end{proof}

% \begin{theorem}\label{theorem:receipt}
% \text{RECEIPT} correctly computes the tip numbers for all $u\in U$.
% \end{theorem}
% \begin{proof}
% Consider an arbitrary vertex $u\in U_i$. From theorem~\ref{theorem:cd}, we can say that $\theta(i)\leq \theta_u< \theta(i+1)$. Let $S=U_1\cup U_2\dots U_{i-1}$ denote the set of vertices peeled before $U_i$ in \texttt{RECEIPT\_CD}. For all vertices $u'\in S$, $\theta_{u'}<\theta_u$ and hence, $S$ will be completely peeled before $u$ in \texttt{BUP} as well. Therefore, from lemmas~\ref{lemma:independence} and \ref{lemma:parallel}, we can say that $\bowtie^{init}_u$ is equal to the support of $u$ in \texttt{BUP} after $S$ is completely peeled.

% We now compare the behavior of \texttt{BUP} after $S$ has been peeled, and \texttt{RECEIPT\_FD} (alg.\ref{alg:fd}).
% By theorem~\ref{theorem:cd}, all the vertices in $U_{i+1}\cup U_{i+2}\dots\cup U_{P}$ have tip numbers strictly greater than $\theta_u$ and will not be peeled before $u$ in \texttt{BUP}. Therefore, sequential peeling of $U_i$ in \texttt{RECEIPT\_FD} follows the same order of vertex peeling as in \texttt{BUP}. Furthermore, the induced subgraph $G_i(U_i, V, E_i)$ used in \texttt{RECEIPT\_FD} has all the butterflies shared between $u$ and other vertices in $U_i$. Therefore, for any vertex $u'$ peeled before $u$, the update $\bowtie_{u', u}$ applied to $\bowtie_u$ will be the same in \texttt{BUP} and \texttt{RECEIPT\_FD}. Thus, the final tip number $\theta_u$ computed by \texttt{RECEIPT\_FD} will be the same as the correct value computed by \texttt{BUP}.
% \end{proof}

\begin{theorem}\label{theorem:receipt}
\text{RECEIPT} correctly computes the tip numbers for all $u\in U$.
\end{theorem}
\begin{proof}
Consider an arbitrary vertex $u\in U_i$. From theorem~\ref{theorem:cd}, $\theta(i)\leq \theta_u< \theta(i+1)$. Let $S=U_1\cup U_2\dots U_{i-1}$ denote the set of vertices peeled before $U_i$ in 
RECEIPT CD. For all vertices $u'\in S$, $\theta_{u'}<\theta_u$ and hence, $S$ will be completely 
peeled before $u$ in \texttt{BUP} as well. \kledit{We now compare  peeling $U_i$ in RECEIPT FD 
(alg.\ref{alg:fd}) to peeling $U_i$ in sequential algorithm \texttt{BUP}, and show that the two output identical tip numbers. It is well known that \texttt{BUP} correctly computes tip numbers~\cite{sariyuce2016peelingarxiv}.}

\kledit{Note that initial support of $u$ in RECEIPT FD i.e. $\bowtie^{init}_u$ is the support of $u$ in RECEIPT
CD after $S$ is peeled (sec.\ref{sec:CD}, lines 6-7 of alg.\ref{alg:cd}). From lemmas \ref{lemma:independence} and \ref{lemma:parallel}. this is equal to the
support of $u$ in \texttt{BUP} after $S$ is peeled. Further, by theorem~\ref{theorem:cd}, any vertex in $U_{i+1}\cup U_{i+2}\dots\cup U_{P}$ has tip number strictly greater than all vertices in $U_i$, and will be peeled after $U_i$ in \texttt{BUP}.}
%\kldelete{Therefore, from lemmas~\ref{lemma:independence} and \ref{lemma:parallel}, we can say that 
%$\bowtie^{init}_u$ is equal to the support of $u$ in \texttt{BUP} after $S$ is completely peeled. }
\kledit{Next, we note that subgraph $G_i$ is induced on subset $U_i$ and entire set $V$. Thus, every butterfly shared between any two vertices in $U_i$ is present in $G_i$. Therefore, for any vertex $u'\in U_i$ peeled before $u$, the update $\bowtie_{u', u}$ computed by \texttt{BUP} and RECEIPT FD will be same.
Hence,  \texttt{BUP} and sequential peeling of $U_i$ in RECEIPT FD apply the same support updates to vertices in $U_i$ and therefore, follow the same order of vertex peeling. Thus, the final tip number $\theta_u$ computed by RECEIPT FD will be the same as that computed by \texttt{BUP}.}
\end{proof}

It is important for a parallel algorithm to be not only scalable, but also computationally efficient. The following theorem shows that for a reasonable upper bound\footnote{In practice, we use $P\ll \frac{\sum_{u\in U}\sum_{v\in N_u}{d_v}}{n\log{n}}$ for large graphs.} on $P$, RECEIPT is at least as efficient as the best sequential tip decomposition algorithm \texttt{BUP}.

\begin{theorem}\label{theorem:complexity}
For $P=\mathcal{O}\left(\frac{\sum_{u\in U}\sum_{v\in N_u}{d_v}}{n\log{n}}\right)$ vertex subsets, RECEIPT is work-efficient with computational complexity of $\mathcal{O}\left(\sum_{u\in U}\sum_{v\in N_u}{d_v}\right)$.
\end{theorem}
\begin{proof}
\underline{\textbf{RECEIPT CD}} -- 
%\kldelete{Butterfly counting used to initialize vertex support in RECEIPT CD (line 2, alg.\ref{alg:cd}) does\\ $\mathcal{O}\left(\sum_{(u,v)\in E}\min\left(d_u, d_v\right)\right)$ work. Range computation is done once per each subset. It first creates a set of unique vertex support values $\bowtie_U$, and accumulates the wedge counts of vertices in the bins corresponding to their support. Next, it sorts $\bowtie_U$ and performs a prefix sum on the bins to compute the $work[\cdot]$ array. Parallel implementations of these primitives perform $\mathcal{O}\left(\abs{U}\log{\abs{U}}\right) = \mathcal{O}\left(n\log{n}\right)$ work.}
\kledit{It initializes the vertex support using ${\mathcal{O}\left(\sum_{(u,v)\in E}\min\left(d_u, d_v\right)\right)}$ complexity per-vertex butterfly counting. Range computation for each subset requires constructing a $\mathcal{O}(\abs{U})$ size hashmap ($work$) whose keys are the unique support values in $\bowtie_U$. In this hashmap, wedge counts of all vertices with support $\bowtie$ are accumulated in value $work[\bowtie]$. Next, $work$ is sorted on the keys and a parallel prefix sum is computed over the values so that the final value $work[\bowtie]$ represents cumulative wedge count of all vertices with support less than or equal to $\bowtie$. Parallel implementations of hashmap generation, sorting and prefix scan perform $\mathcal{O}\left(\abs{U}\log{\abs{U}}\right) = \mathcal{O}\left(n\log{n}\right)$ work. Computing scaling factor $s_i$ for adaptive range determination requires aggregating wedge count of vertices in $U_i$, contributing $\mathcal{O}\left(\abs{U} \right) = \mathcal{O}\left(n\right)$ work over all subsets.}

Constructing $activeSet$ for first peeling iteration of each subset requires an $\mathcal{O}(n)$ complexity parallel filtering on $\bowtie_U$. Subsequent iterations construct $activeSet$ by tracking support updates doing $\mathcal{O}(1)$ work per update. 
%\kldelete{The \texttt{\textsc{update}} function is called once for every vertex in $u\in U$ and does $\mathcal{O}\left(\sum_{v\in N_u}{d_v} \right)$ work.}
\kledit{For every vertex $u\in U$, the \texttt{\textsc{update}} routine is called once when $u$ is peeled, and it traverses at most $\sum_{v\in N_u}{d_v}$ wedges. At most one support update is generated per wedge, resulting in $\mathcal{O}(1)$ work per wedge \kledit{(RECEIPT CD stores vertex supports in an array $\bowtie_U$}).} Thus, the complexity of RECEIPT CD is $\mathcal{O}\left(\sum_{u\in U}\sum_{v\in N_u}{d_v} + Pn\log{n}\right)$.

\noindent\underline{\textbf{RECEIPT FD}} -- For every subset $U_i$, alg.\ref{alg:fd} creates an induced subgraph $G_i(U_i, V_i, E_i)$ in parallel. This requires $\mathcal{O}\left(n + \abs{E_i}\right)$ work for subset $U_i$ and $\mathcal{O}\left(Pn + m\right)$ work over all $P$ subsets. 
%The \texttt{\textsc{update}} routine is called once for every vertex $u\in U$. 
When $u\in U_i$ is peeled, the corresponding call to \texttt{\textsc{update}} explores wedges in the induced subgraph $G_i$ and generates support updates to other vertices in $U_i$. These are a subset of the wedges and support updates generated when $u$ was peeled in RECEIPT CD.
%It explores the wedges in $G_i$ starting at $u$, which are a subset of wedges in $G$ with endpoint $u$. 
\kledit{A fibonacci heap can be used to extract minimum support vertex ($\mathcal{O}\left(\log{n}\right)$ work per vertex), and enable constant complexity updates ($\mathcal{O}\left(\sum_{v\in N_u}{d_v} \right)$ work in \texttt{\textsc{update}} call for $u$). Thus, the complexity of RECEIPT FD is $\mathcal{O}\left(\sum_{u\in U}\sum_{v\in N_u}{d_v} + Pn + n\log{n}\right)$.}

Combining both the steps, the total work done by RECEIPT is $\mathcal{O}\left(\sum_{u\in U}\sum_{v\in N_u}{d_v} + Pn\log{n}\right)=\mathcal{O}\left(\sum_{u\in U}\sum_{v\in N_u}{d_v}\right)$, if $P=\mathcal{O}\left(\frac{\sum_{u\in U}\sum_{v\in N_u}{d_v}}{n\log{n}}\right)$. This is the same as sequential \texttt{BUP} and hence, RECEIPT is work-efficient.
\end{proof}

% \begin{itemize}
%     \item Prove correctness.
%     \item Derive theoretical complexity. Practically, $P\*V \ll \sum_u\sum_{v\in N(u)}deg(v)$, then theoretically work-efficient.
% \end{itemize}

%% file: optimizations.tex
\section{OPTIMIZATIONS}\label{sec:optimizations}
Despite the parallelism potential, RECEIPT may take hours or days to process large graphs such as \textit{TrU}, that contain hundreds of trillions of wedges (sec.\ref{sec:introduction}). To this purpose, we develop novel optimizations that \textit{exploit the properties of} RECEIPT \textit{to dramatically improve its computational efficiency} in practice, making it feasible to decompose graphs like \textit{TrU} in minutes (objective 2, sec.\ref{sec:challenges}).

\subsection{Hybrid Update Computation (HUC)}\label{sec:hybrid}
The complexity of RECEIPT is dominated by wedge traversal done during peeling in RECEIPT CD. In order to reduce this traversal, we exploit the following insights about the behavior of counting and peeling algorithms:

\begin{itemize}
    \item Butterfly counting is computationally efficient (low complexity) and easily parallelizable (sec.\ref{sec:counting}).
    \item Some peeling iterations in RECEIPT CD may peel a large number of \looseness=-1vertices.
\end{itemize}

%\klcomment{This formulation looks a bit weird and unnecessarily introduces notations.}
\kledit{Given a vertex set ($activeSet$) to peel, we compute the cost of peeling $C_{peel}$ as $\sum_{u\in activeSet}{\sum_{v\in N_u} d_v}$, which is the total wedge count of vertices in $activeSet$. However, the cost of re-counting butterflies $C_{rcnt}$, is computed as $\sum_{(u,v)\in E}\min\left(d_u, d_v\right)$ which represents the bound on wedge traversal of counting (sec.\ref{sec:background}). Thus, if $C_{peel}$ 
exceeds $C_{rcnt}$, we \textit{re-compute butterflies} for all remaining vertices in $U$ instead
of computing support updates obtained by peeling $activeSet$.} 
% \kldelete{Thus, if the total 
% wedges with end-points as the vertices to be peeled ($activeSet$) exceeds the amount of wedges 
% traversed during counting ($\sum_{(u,v)\in E}\min\left(d_u, d_v\right)$), we \textbf{replace 
% peeling with re-counting}. We flag the vertices in $activeSet$ as deleted and run a modified 
% version of \texttt{pvBcnt} (alg.\ref{alg:counting}) that skips over these deleted vertices and 
% the ones peeled in earlier iterations}
%\kldelete{\footnote{To generate the $activeSet$ for the peeling iteration that immediately follows the re-count, we scan all vertices in $U$ and check if their new butterfly count differs from their previous support.}}. 
This optimization is denoted Hybrid Update Computation (HUC). 
% Note that during a \textit{re-count}, we only compute butterflies for vertices in $U$ (and don't update $\bowtie_V$). 

%\kldelete{While we enable HUC in RECEIPT FD as well, it fetches little advantage there as}
We note that compared to RECEIPT CD, HUC is relatively less beneficial for RECEIPT FD because:
\begin{enumerate*}[label=(\alph*)]
\item the induced subgraphs have significantly less wedges than the original graph, and 
\item few vertices are typically deleted per peeling iteration in RECEIPT FD.
\end{enumerate*}
Further, re-counting for subset $U_i$ in RECEIPT FD must account for butterflies with vertices in $\cup_{j>i}{\{U_j\}}$. This external contribution for a vertex $u$ can be computed by deducting the butterfly count of $u$ within $G_i$, from $\bowtie^{init}_u$.
%must be computed once per subgraph 
%if re-counting is to be employed in any iteration.

\subsection{Dynamic Graph Maintenance (DGM)}\label{sec:dynamic}
We note that after a vertex $u$ is peeled in RECEIPT CD or RECEIPT FD, it is excluded from future computation in the respective step. However, the graph data structure (adjacency list/Compressed Sparse Row) still contains edges to $u$ interleaved with other edges. Consequently, wedges incident on $u$ are still explored after $u$ is peeled. To prevent such wasteful exploration, we \textit{periodically update the 
data structures} to remove edges incident on peeled vertices. We denote this optimization as Dynamic Graph Maintenance (DGM).

%this is essentially a parallel compact operation on every adjacency list that can be done in 
%time linear to the size of a list. Hence, O(m) work
The cost of DGM is determined by the work done in parallel compaction of all adjacency lists, which grows linearly with the number of edges in the graph. 
% \kldelete{Removing edges of peeled vertices from the graph data structure requires $O(m)$ parallel work. Clearly, updating the data structure too frequently can result in large overhead.} 
Therefore, if the adjacency lists are updated only after $m$ wedges have been traversed since previous update, DGM will not 
%To reduce the associated overhead, we update the adjacency lists only after $m$ wedges have been traversed since last update. Thus, DGM does not 
alter the theoretical complexity of RECEIPT and pose negligible practical overhead.

%% file: experiments.tex
\section{Experiments}\label{sec:exp}
\subsection{Setup}\label{sec:setup}
We conduct the experiments on a 36 core dual-socket linux server with two Intel Xeon E5-2695 v4 processors@ 2.1GHz and 1TB DRAM. All algorithms are implemented in C++-14 and are 
compiled using G++ 9.1.0 with the -O3 optimization flag. We use OpenMP v4.5 for multithreading. \vspace{1.5mm}

\noindent\textbf{\textit{Baselines}}: We compare RECEIPT against the sequential \texttt{BUP} algorithm (alg.\ref{alg:bottomup}) and its parallel variant \texttt{ParB}~\cite{shiParbutterfly}. \texttt{ParB} resembles P\textsc{ar}B\textsc{utterfly} with BATCH mode peeling\footnote{We were unable to verify the correctness of tip numbers generated by public code for  P\textsc{ar}B\textsc{utterfly} and hence, implemented it ourselves for comparison.}. \texttt{ParB} uses the bucketing structure of Julienne~\cite{julienne} with $128$ buckets as given in \cite{shiParbutterfly}.\vspace{1mm}

\noindent\textbf{\textit{Datasets}}: We use six unweighted bipartite graphs obtained from the KOBLENZ  collection~\cite{konect}. To the best of our knowledge, these are some of the largest publicly available bipartite datasets. \kledit{Within each graph, number of wedges with endpoints in one vertex set can be significantly different than the other, as can be seen from $\wedge^{BUP}$ in table~\ref{table:performance}. We label the vertex set with higher number of wedges as $U$ and the other as $V$}, and accordingly suffix "U" or "V" to the dataset name to identify which vertex set is \looseness=-1decomposed.
%distinguish whether the tip-decomposition is carried out for $U$ or $V$ vertex set, respectively. 

\begin{table*}[htbp]
\caption{Bipartite Datasets for evaluation with the corresponding number of butterflies ($\bowtie_G$) and wedges ($\wedge_G$) in billions, and maximum tip numbers for $U$ ($\theta^{max}_U$) and $V$ ($\theta^{max}_V$). \kledit{$d_U$ and $d_V$ denote the average degree of vertices in $U$ and $V$, respectively.}}
\label{table:datasets}
\resizebox{\linewidth}{!}{%
\begin{tabular}{|c|c|c|c|c|c|c|c|c|c|}
\hline
\textbf{Dataset} & \textbf{Description}                           & $\mathbf{\abs{U}}$ & $\mathbf{\abs{V}}$ & $\mathbf{\abs{E}}$ & \kledit{$\mathbf{d_U\ /\ d_V}$} & $\mathbf{\boldsymbol{\bowtie}_G}$(in B) & $\mathbf{\boldsymbol{\wedge}_G}$(in B) & $\mathbf{\boldsymbol{\theta}^{max}_U}$ &$\mathbf{\boldsymbol{\theta}^{max}_V}$  \\ \hline\hline
ItU, ItV         & Pages and editors from Italian Wikipedia       & 2,255,875                                          & 137,693                        & 12,644,802        & \kledit{5.6 / 91.8}                                 & 298 

                                                                                                                 & 361                          &  1,555,462 &  5,328,302,365                                                                                          \\ \hline
DeU, DeV         & Users and tags from www.delicious.com          & 4,512,099                                          & 833,081                                            & 81,989,133 & \kledit{18.2 / 98.4}                                         & 26,683                                                                                                                & 1,446 &  936,468,800 &  91,968,444,615 
                                                                                                                \\ \hline
OrU, OrV         & Users' group memberships in Orkut         & 2,783,196                                          & 8,730,857                                          & 327,037,487    & \kledit{117.5 / 37.5}                                    & 22,131                                                                                                                & 2,528 &  88,812,453 &  29,285,249,823                                                                                                               \\ \hline
LjU, LjV         & Users' group memberships in Livejournal   & 3,201,203                                          & 7,489,073                                          & 112,307,385   & \kledit{35.1 / 15}                                     & 3,297                                                                                                                 & 2,703 &  4,670,317 &  82,785,273,931 
                                                                                                               \\ \hline
EnU, EnV         & Pages and editors from English Wikipedia       & 21,504,191                                         & 3,819,691                                          & 122,075,170  & \kledit{5.7 / 32}                                      & 2,036                                                                                                                 & 6,299 &  37,217,466 &  96,241,348,356 

                                                                                                                \\ \hline
TrU, TrV         & Internet domains and trackers in them & 27,665,730                                         & 12,756,244                                         & 140,613,762 & \kledit{5.1 / 11}                                        & 20,068                                                                                                                & 106,441 &  18,667,660,476 &  3,030,765,085,153 
                                                                                                             \\ \hline
\end{tabular}
}
\end{table*}

\begin{figure}[htbp]
    \centering
\includegraphics[width=0.9\linewidth]{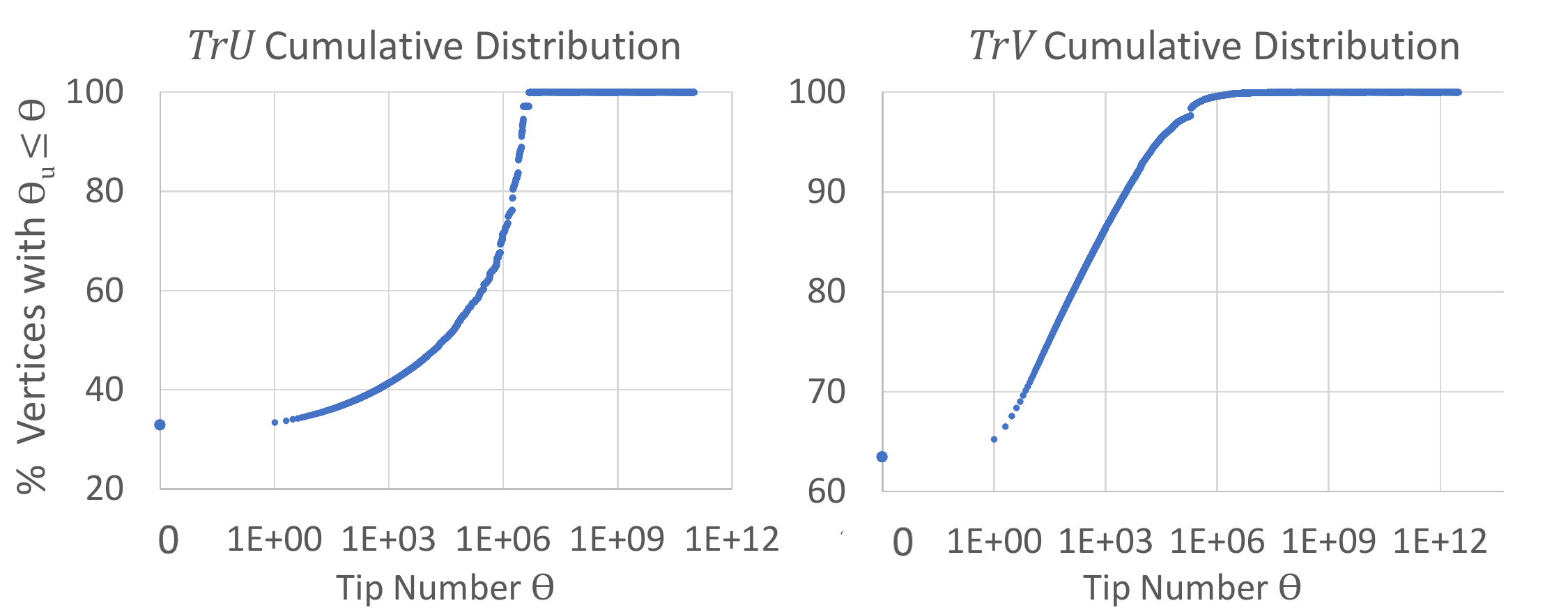}
\caption{Tip number distribution in Trackers graph -- y-axis shows percentage vertices with tip number less than abscissa.\vspace{-2mm}}
    \label{fig:tipNumber}
\end{figure}

Note that the maximum tip numbers are extremely high because of few high degree vertices sharing a large number of common neighbors. However, most of the vertex' tip numbers lie in a much smaller range as shown in fig.\ref{fig:tipNumber}. For example, $99.98$\% vertices in $TrU$ have tip number $<5$M ($0.027$\% of \looseness=-1maximum).\vspace{1mm}

\noindent\textbf{\textit{Implementation Details}}: Unless otherwise mentioned, the parallel algorithms use all $36$ threads for execution\footnote{We did not observe any benefits from enabling the $2$-way hyperthreading.} and RECEIPT includes all workload optimizations discussed in sec.\ref{sec:optimizations}. In RECEIPT FD and sequential \texttt{BUP}, we use a $k$-way min-heap for efficient retrieval of minimum support vertices. We found it to be faster in practice than the bucketing structure of \cite{sariyucePeeling} or fibonacci heaps.

To analyze the impact of number of vertex subsets on RECEIPT's performance, we vary $P$ from $50$ to $500$. Fig.\ref{fig:parts} reports the effect of $P$ on some large datasets that showed significant performance variation. \kledit{Typically, RECEIPT CD dominates the overall execution time because of the larger number of wedges traversed compared to RECEIPT FD}. The performance of RECEIPT CD improves with a decrease in $P$ because of reduced synchronizations. However, a small value of $P$ reduces parallelism in RECEIPT FD and makes the induced subgraphs larger. \kledit{Consequently, for $P\leq 100$, we observed that RECEIPT FD became the bottleneck for some large graphs such as $TrU$ and $LjU$}. For all the datasets shown in fig.\ref{fig:parts}, execution slows down when $P$ is increased beyond $150$. Based on these observations, we use $P=150$ for all datasets, in rest of the experiments.

\begin{figure}[htbp]
    \centering
\includegraphics[width=0.8\linewidth]{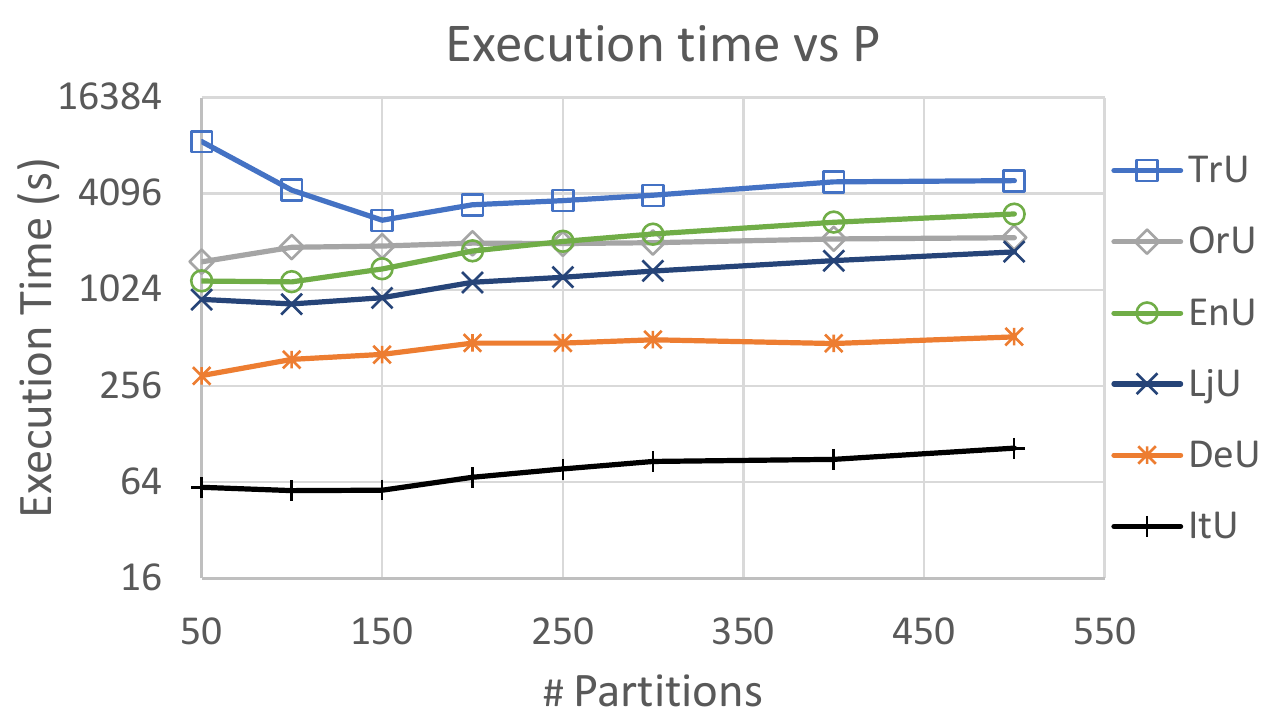}
\caption{Execution time of RECEIPT vs $P$.\vspace{-2mm}}
    \label{fig:parts}
\end{figure}

\subsection{Results}\label{sec:results}
\subsubsection{Comparison with Baselines}\label{sec:baselines}

\begin{table*}[]
\caption{Comparing execution time ($t$), \# wedges traversed ($\bigwedge$) and \# synchronization rounds ($\rho$) of RECEIPT and baseline algorithms. \texttt{ParB} traverses the same \# wedges as \texttt{BUP} and has missing entries due to out-of-memory error. \texttt{pvBcnt} denotes per-vertex butterfly counting and $t=\infty$ denotes execution did not finish in $10$ days.}
\label{table:performance}
\resizebox{0.925\linewidth}{!}{%
\begin{tabular}{cc|c|c|c|c|c|c|c|c|c|c|c|c|}
\cline{3-14}
\multicolumn{1}{c}{}                                                                              & \multicolumn{1}{c|}{\textbf{}} & \multicolumn{1}{c|}{\textbf{ItU}}    & \multicolumn{1}{c|}{\textbf{ItV}}    & \multicolumn{1}{c|}{\textbf{DeU}}     & \multicolumn{1}{c|}{\textbf{DeV}}    & \multicolumn{1}{c|}{\textbf{OrU}}     & \multicolumn{1}{c|}{\textbf{OrV}}     & \multicolumn{1}{c|}{\textbf{LjU}}     & \multicolumn{1}{c|}{\textbf{LjV}}    & \multicolumn{1}{c|}{\textbf{EnU}}     & \multicolumn{1}{c|}{\textbf{EnV}}    & \multicolumn{1}{c|}{\textbf{TrU}}     & \multicolumn{1}{c|}{\textbf{TrV}}     \\ \hline\hline
\multicolumn{1}{|c|}{}                                                                            & \kledit{\texttt{pvBcnt}}                            & \kledit{0.3}                                & \kledit{0.3}                                  & \kledit{8.3}                                & \kledit{8.3}                                  & \kledit{45.6}                                & \kledit{45.6}                                 & \kledit{5.1}                                & \kledit{5.1}                                  & \kledit{6.9}                               & \kledit{6.9}                                  & \kledit{7.8}                                     & \kledit{7.8}                                 \\ \cline{2-14} 
\multicolumn{1}{|c|}{}                                                                            & \texttt{BUP}                            & 3,849                                & 8.4                                  & 12,260                                & 428                                  & 39,079                                & 2,297                                 & 67,588                                & 200                                  & 111,777                               & 281                                  & $\infty$                                     & 5,711                                 \\ \cline{2-14} 
\multicolumn{1}{|c|}{}                                                                            & \texttt{ParB}                           & 3,677                                & 8.1                                  & -                                     & 377.7                                & -                                     & 1,510                                 & -                                     & 132.5                                & -                                     & 198                                  & -                                     & 3,524                                 \\ \cline{2-14} 
\multicolumn{1}{|c|}{\multirow{-4}{*}{\begin{tabular}[c]{@{}c@{}}$\mathbf{t}$(s)\end{tabular}}}         & \texttt{RECEIPT}                        & {\color[HTML]{036400} \textbf{56.8}} & {\color[HTML]{036400} \textbf{3.1}}  & {\color[HTML]{036400} \textbf{402.4}} & {\color[HTML]{036400} \textbf{32.4}} & {\color[HTML]{036400} \textbf{1,865}} & {\color[HTML]{036400} \textbf{136}}   & {\color[HTML]{036400} \textbf{911.1}} & {\color[HTML]{036400} \textbf{23.7}} & {\color[HTML]{036400} \textbf{1,383}} & {\color[HTML]{036400} \textbf{31.1}} & {\color[HTML]{036400} \textbf{2,784}} & {\color[HTML]{036400} \textbf{530.6}} \\ \hline\hline
\multicolumn{1}{|c|}{}                                                                            & \kledit{\texttt{pvBcnt}}                            & \kledit{0.19}                                & \kledit{0.19}                                  & \kledit{20.3}                                & \kledit{20.3}                                  & \kledit{74.8}                                & \kledit{74.8}                                 & \kledit{4.7}                                & \kledit{4.7}                                  & \kledit{6.3}                               & \kledit{6.3}                                  & \kledit{6.3}                                     & \kledit{6.3}                                 \\ \cline{2-14}
\multicolumn{1}{|c|}{}                                                                            & \texttt{BUP}                            & 723                                & 0.57 & 2,861                                 & 70.1 & 4,975                                 & 231.4 & 5,403                                 & 14.3  & 12,583                                & 29.6                                 & 211,156                               & 1,740                                 \\ \cline{2-14} 
\multicolumn{1}{|c|}{\multirow{-3}{*}{\begin{tabular}[c]{@{}c@{}}$\mathbf{\boldsymbol{\bigwedge}}$\\ (billions)\end{tabular}}} & \texttt{RECEIPT}                        & {\color[HTML]{036400} \textbf{71}}   & {\color[HTML]{036400} \textbf{0.56}}                                 & {\color[HTML]{036400} \textbf{1,503}} & {\color[HTML]{036400} \textbf{51.3}}                                 & {\color[HTML]{036400} \textbf{2,728}} & {\color[HTML]{036400} \textbf{170.4}}                                 & {\color[HTML]{036400} \textbf{1,003}} & {\color[HTML]{036400} \textbf{11.7}}                                 & {\color[HTML]{036400} \textbf{2,414}} & {\color[HTML]{036400} \textbf{22.2}} & {\color[HTML]{036400} \textbf{3,298}} & {\color[HTML]{036400} \textbf{658.1}} \\ \hline\hline
\multicolumn{1}{|c|}{}                                                                            & \texttt{ParB}                           & 377,904                              & 10,054                               & 670,189                               & 127,328                              & 1,136,129                             & 334,064                               & 1,479,495                             & 83,423                               & 1,512,922                             & 83,800                               & 1,476,015                             & 342,672                               \\ \cline{2-14} 
\multicolumn{1}{|c|}{\multirow{-2}{*}{$\boldsymbol{\rho}$}}                                                    & \texttt{RECEIPT}                        & {\color[HTML]{036400} \textbf{967}}  & {\color[HTML]{036400} \textbf{280}}  & {\color[HTML]{036400} \textbf{1113}}  & {\color[HTML]{036400} \textbf{406}}  & {\color[HTML]{036400} \textbf{1,160}} & {\color[HTML]{036400} \textbf{639}}   & {\color[HTML]{036400} \textbf{1,477}} & {\color[HTML]{036400} \textbf{456}}  & {\color[HTML]{036400} \textbf{1,724}} & {\color[HTML]{036400} \textbf{453}}  & {\color[HTML]{036400} \textbf{1,335}} & {\color[HTML]{036400} \textbf{1,381}} \\ \hline
\end{tabular}
}
\end{table*}

Table~\ref{table:performance} shows a detailed comparison of various tip decomposition algorithms. Note that $\rho$ represents the number of synchronization rounds. Each round consists of one iteration of peeling all vertices with minimum support (or support in minimum range for RECEIPT)\footnote{Wedge traversal by \texttt{BUP} can be computed without executing alg.\ref{alg:bottomup}, by simply aggregating the $2$-hop neighborhood size of vertices in $U$ or $V$. A given wedge can be traversed twice. Similarly, $\rho$ for \texttt{ParB} can be computed from a modified version of RECEIPT FD where we retrieve all vertices with minimum support in a single iteration.}. Each round involves multiple (but constant) thread synchronizations, for example, once after computing the list of vertices to peel, once after computing the updates etc. RECEIPT FD does not incur any synchronization as the threads in alg.\ref{alg:fd} only synchronize once at the end of computation. For the rest of this section, we will use the term \textit{large datasets} to refer to a graph having large number of wedges with endpoints on the vertex set being peeled such as \textit{ItU, OrU} etc. \vspace{1mm}

\noindent\textbf{\textit{Execution Time}} ($\mathbf{t}$): With up to $\mathbf{80.8\boldsymbol{\times}}$ and $\mathbf{64.7\boldsymbol{\times}}$ speedup over \texttt{BUP} and \texttt{ParB}, respectively, RECEIPT is radically faster than both baselines, for \textit{all} datasets. The speedups are typically higher for large datasets that offer large amount of computation to parallelize and bigger benefits from workload optimizations(sec.\ref{sec:optimizations}). Only RECEIPT is able to successfully tip decompose \textit{TrU}. Contrarily, \texttt{ParB} achieves a maximum $1.6\times$ speedup compared to sequential \texttt{BUP} for \textit{TrV}.\vspace{1mm}

\noindent\textbf{\textit{Wedges Traversed}} ($\bigwedge$): \kledit{For \textit{all} datasets, RECEIPT traverses fewer wedges than the existing baselines. RECEIPT's built-in optimizations achieve up to $\mathbf{64\boldsymbol{\times}}$ reduction in wedges traversed. This enables RECEIPT to decompose large datasets \textit{EnU} and \textit{TrU} in few minutes, unlike baselines that take few days for the same.\vspace{1mm}}
% \kldelete{Further, the rate of traversal is much higher in RECEIPT compared to baselines because of parallelism and faster support updates. 
% In RECEIPT CD, support updates are applied to an array as opposed to the heap or bucketing structures in baselines. RECEIPT FD relatively traverses much less wedges, incurs fewer support updates and operates on a much smaller heap than \texttt{BUP}.}

\noindent\textbf{\textit{Synchronization}} ($\rho$): In comparison with \texttt{ParB}, RECEIPT reduces the amount of thread synchronization by up to $\mathbf{1105\boldsymbol{\times}}$. This is primarily because RECEIPT CD peels vertices with a broad range of support in every iteration. This drastic reduction in $\rho$ is the primary contributor to RECEIPT's parallel efficiency.

% x-axis is datasets
% y-axis has 3 parts: execution time (T), synch rounds (rho), wedges traversed (wedge)
% For each part, 2 or 3 rows
% \begin{itemize}
%     \item Compare execution time
%     \item Compare \# synchronization rounds
% \end{itemize}

\subsubsection{Effect of Workload Optimizations}\label{sec:optExp}
Fig.\ref{fig:optWedges} and \ref{fig:optTime} show the effect of HUC and DGM optimizations (sec.\ref{sec:hybrid} and \ref{sec:dynamic}) on the performance of RECEIPT. The execution time closely follows the trend in wedge \looseness=-1workload.

\begin{figure}[htbp]
    \centering
\includegraphics[width=0.9\linewidth]{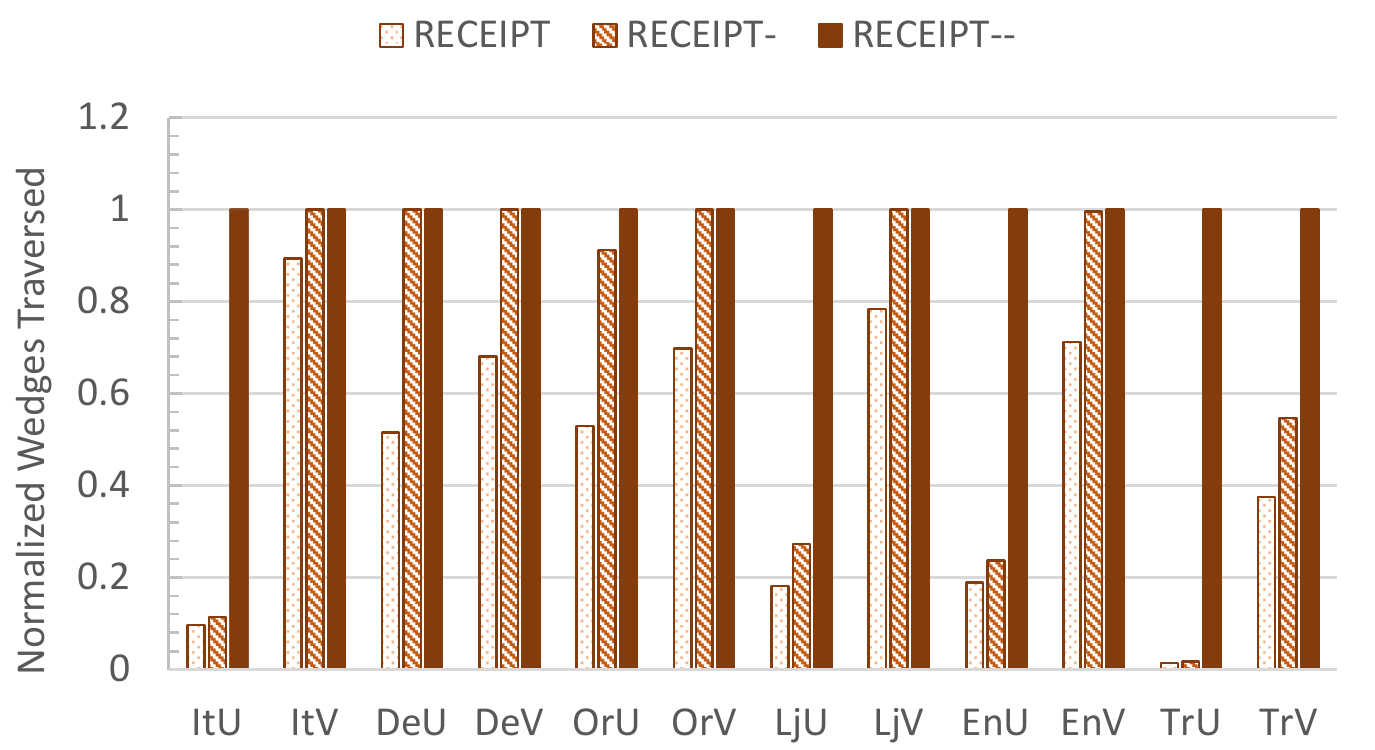}
\caption{Effect of workload optimizations on wedge traversal (normalized with respect to RECEIPT-{}-). RECEIPT- is RECEIPT without DGM. RECEIPT-{}- is RECEIPT without DGM and HUC.\vspace{-1mm}}
    \label{fig:optWedges}
\end{figure}

\begin{figure}[htbp]
    \centering
\includegraphics[width=0.9\linewidth]{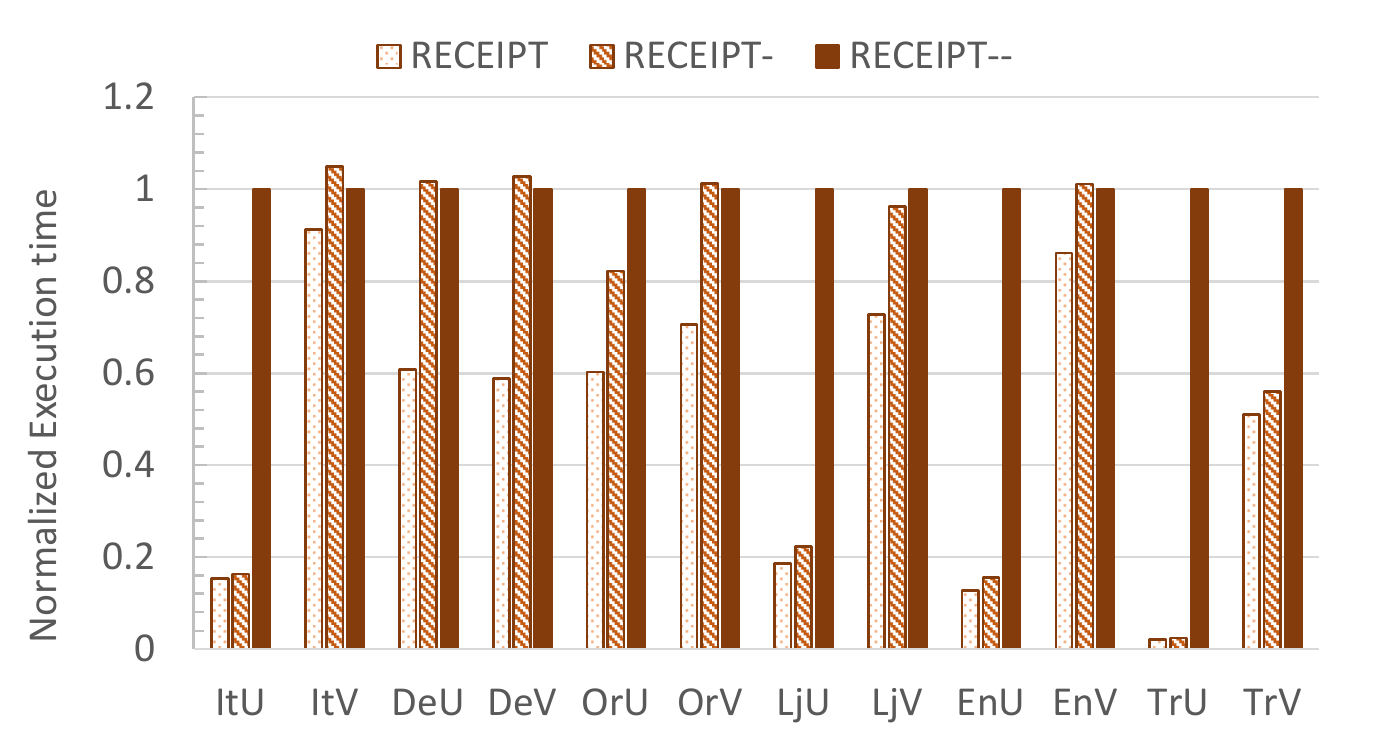}
\caption{Effect of workload optimizations on execution time (normalized with respect to RECEIPT-{}-). RECEIPT- is RECEIPT without DGM. RECEIPT-{}- is RECEIPT without DGM and HUC.\vspace{-2mm}}
    \label{fig:optTime}
\end{figure}

\kledit{HUC reduces wedge traversal by opting to re-count butterflies instead of peeling, in 
selective iterations. A comparison of the the work required
for peeling vertices in $U$ versus counting all butterflies can be an indicator of the benefits of HUC.
Therefore, we define a ratio $r=\frac{\wedge^{peel}}{\wedge^{cnt}}$, where $\wedge^{peel}$ and $\wedge^{cnt}$ denote the 
number of wedges traversed for peeling all vertices in $U$ ($\wedge^{BUP}-\wedge^{pvBcnt}$ in table~\ref{table:performance}) and for butterfly counting 
($\wedge^{pvBcnt}$ in table~\ref{table:performance}), respectively. Datasets with large value of $r$ (for example \textit{ItU, LjU, EnU} and \textit{TrU} with $r>1000$) benefit heavily from HUC optimization, since peeling in high workload iterations is replaced by re-counting, which dramatically reduces wedge traversal. Especially for \textit{TrU} ($r>33,500$), HUC enables \textbf{$57\times$} and \textbf{$42\times$} reduction in wedge traversal and execution time, respectively. Contrarily, in datasets with small value of $r$ (\textit{ItV, DeV, OrV, LjV and EnV} with $r<5$ due to low $\wedge^{peel}$), none of the iterations utilize re-counting. Consequently, performance of RECEIPT- and RECEIPT-{}- is similar for these datasets.}

% \kldelete{
% HUC creates an upper bound for the wedges explored per 
% peeling iteration, which is independent of the wedges 
% incident on vertices being peeled. Thus, for large 
% datasets ItU, LjU, EnU and TrU, it dramatically decreases 
% the total work. The skewed wedge distribution across 
% vertices in these datasets generates some very high 
% workload iterations, where RECEIPT adopts re-counting 
% instead of peeling (sec.\ref{sec:optimizations}). As shown
% in fig.\ref{fig:optTime}, \textit{TrU} could not be 
% processed in a day without HUC. For small datasets and some large datasets, HUC is under utilized and provides negligible benefits. As an example, only $3$ out of $1160$ RECEIPT CD peeling iterations in \textit{OrU} use re-counting.}

DGM can potentially half the wedge workload since each wedge has two endpoints, but needs to be traversed only by the vertex that gets peeled first. However, the reduction is less than $2\times$ in practice, because for many wedges, both endpoints get peeled between consecutive DGM data structure updates.
%DGM does not update the graph after every peeling iteration. Consequently, for several wedges, both the endpoints get peeled before a DGM update filters the graph and deletes those wedges. 
It achieves $1.41\times$ and $1.29\times$ average reduction in wedges and execution time, \looseness=-1respectively.

% Compare existing vs (existing - (parallel scheduling + successive graph size reduction)) vs (existing - (parallel scheduling + successive graph size reduction + re-count/peel hybrid))

% two figures for execution time and wedges explored (can use select few datasets)

\subsubsection{RECEIPT Breakup}\label{sec:breakup}
In this section, we analyze the work and execution time contribution of each step of RECEIPT individually. We further split RECEIPT CD (alg.\ref{alg:bottomup}) into a peeling step (denoted as RECEIPT CD), and a per-vertex butterfly counting (\texttt{pvBcnt}) step used for support initialization. 

\begin{figure}[htbp]
    \centering
\includegraphics[width=0.9\linewidth]{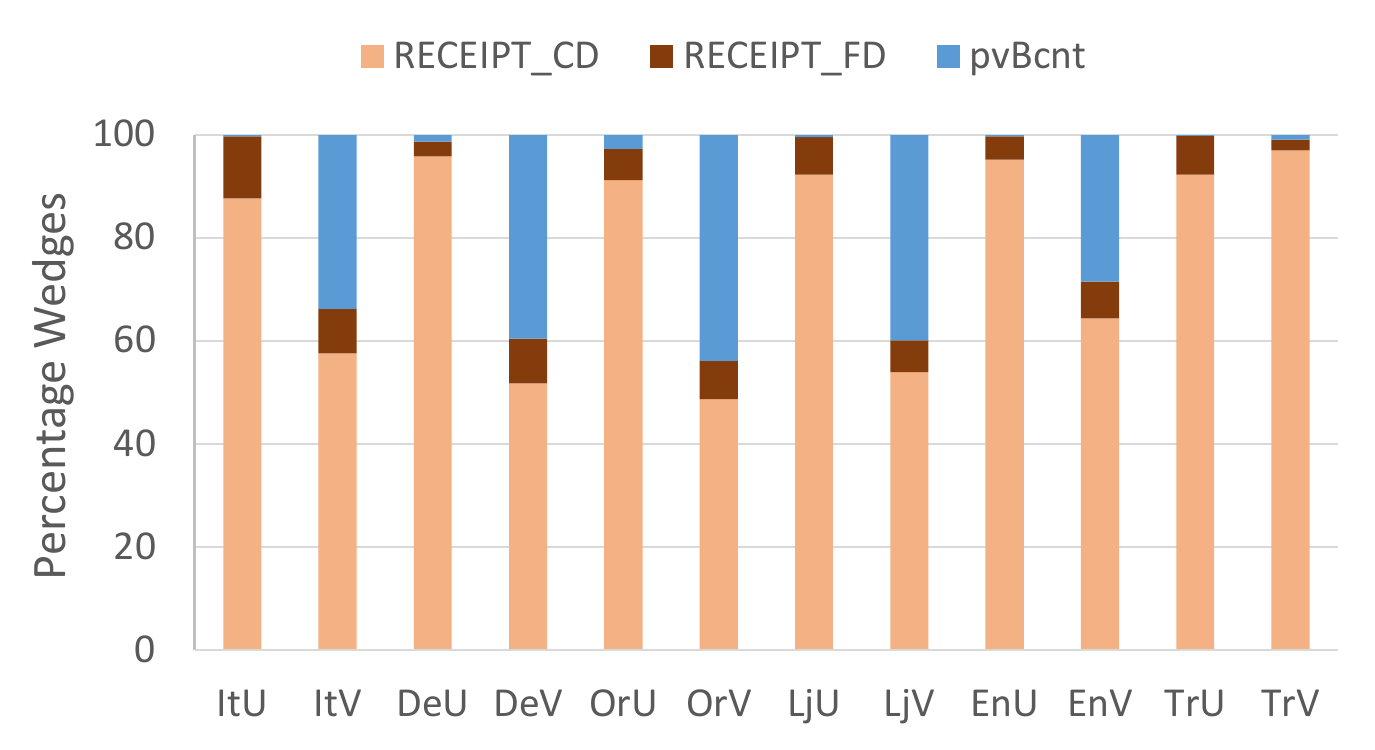}
\caption{Breakup of wedges traversed in RECEIPT\vspace{-1.5mm}}
    \label{fig:wb}
\end{figure}

Fig.\ref{fig:wb} shows a breakdown of the wedges \kledit{traversed during different steps as a percentage of total wedges traversed by RECEIPT}. As discussed in sec.\ref{sec:receipt}, RECEIPT CD traverses significantly more wedges than RECEIPT FD. For all the datasets, RECEIPT FD incurs $<15$\% of the total wedge traversal. Note that for a given graph, the number of wedges explored by \texttt{pvBcnt} is independent of the vertex set being peeled. For example, \textit{ItU} and \textit{ItV} both traverse $188$ million wedges during \texttt{pvBcnt}. However, it's percentage contribution depends on the total wedges explored during the entire tip decomposition, which can vary significantly between $U$ and $V$ vertex sets (table~\ref{table:performance}).

\begin{figure}[htbp]
    \centering
\includegraphics[width=0.9\linewidth]{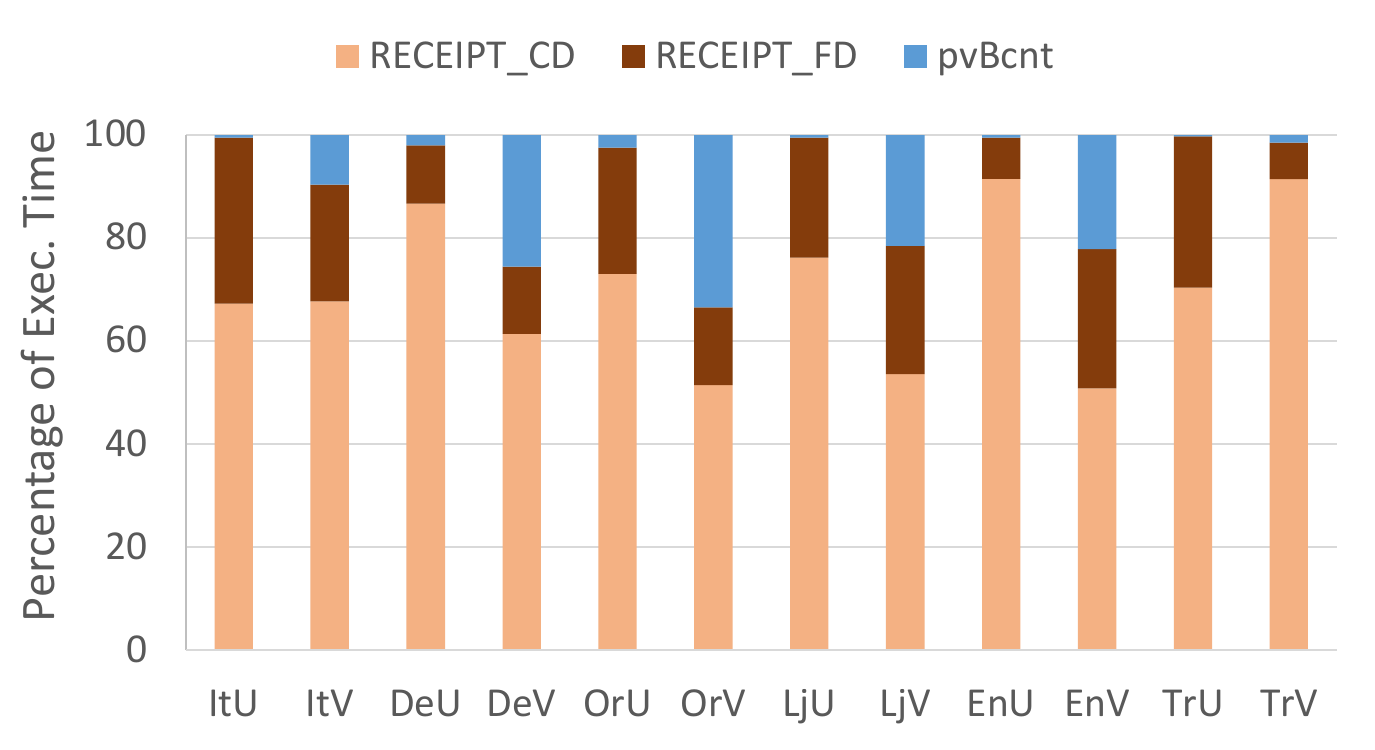}
\caption{Breakup of execution time of RECEIPT\vspace{-1.5mm}}
    \label{fig:tb}
\end{figure}
Fig.\ref{fig:tb} shows a breakdown of the execution time of different steps \kledit{as a percentage of the total execution time of RECEIPT}. Intuitively, the step with most workload i.e. RECEIPT CD, has the largest contribution  ($>50\%$ of the total execution time for all datasets). \kledit{In most datasets with a small value of $r=\frac{\wedge^{peel}}{\wedge^{cnt}}$ (\textit{ItV, DeV, OrV, LjV} and \textit{EnV}), even \texttt{pvBcnt} has a significant share of the execution time. Note that for some datasets, contribution of RECEIPT FD to the execution time is more than its contribution to the wedge traversal. This is due to the following reasons:
\begin{enumerate*}[label=(\alph*)]
\item RECEIPT FD has computational overheads other than wedge exploration, such as creating induced subgraphs and applying support updates to a heap, and 
\item Lower parallel scalability compared to counting and RECEIPT CD (sec.\ref{sec:scalability}).
\end{enumerate*}
Still, RECEIPT FD consumes $<25\%$ of the overall execution time for almost all datasets.
}

\subsubsection{Parallel Scalability}\label{sec:scalability}
\begin{figure}[htbp]
    \centering
\includegraphics[width=0.85\linewidth]{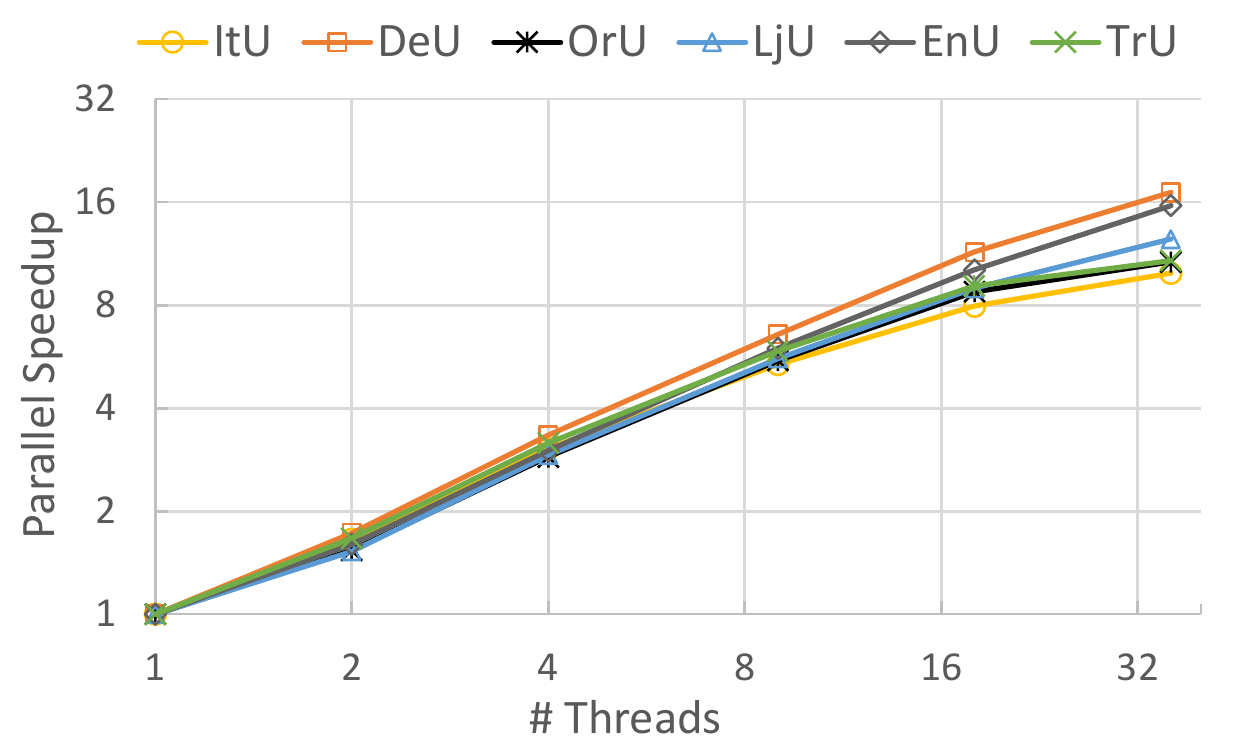}
\caption{Parallel Speedup of RECEIPT when peeling set $U$\vspace{-2.5mm}}
    \label{fig:scaleU}
\end{figure}

\begin{figure}[htbp]
    \centering
\includegraphics[width=0.85\linewidth]{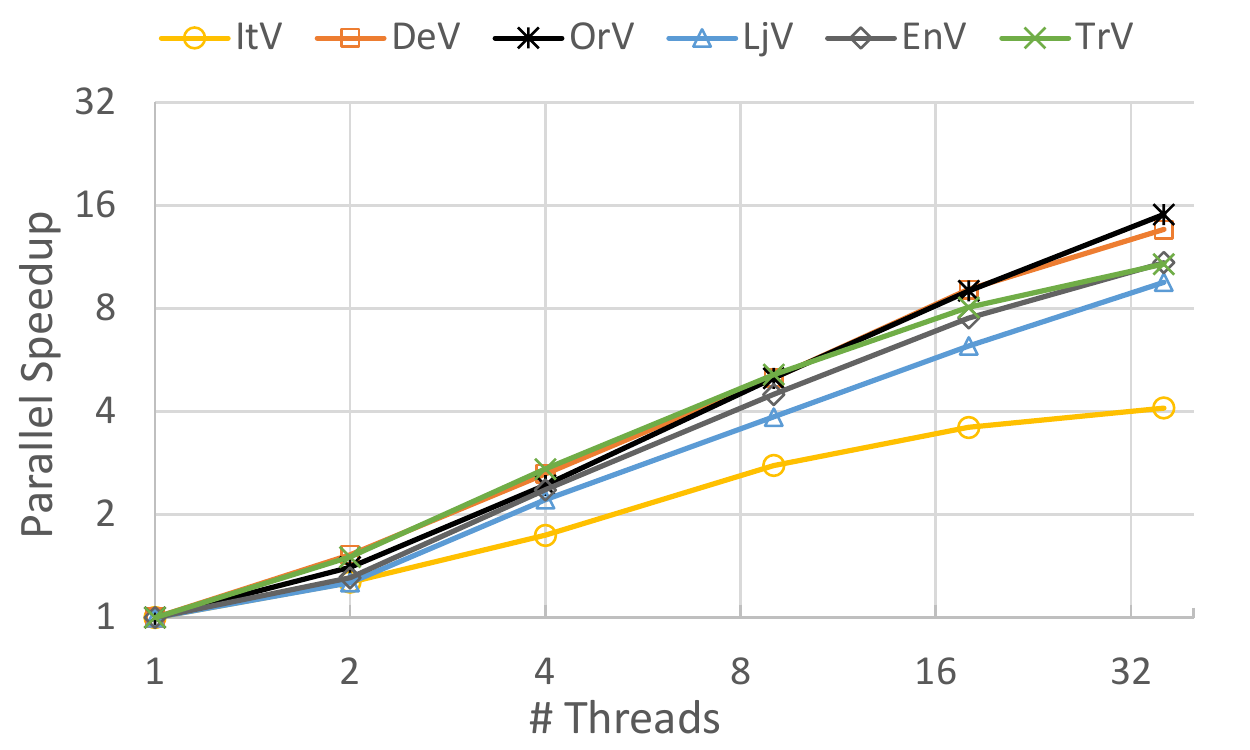}
\caption{Parallel Speedup of RECEIPT when peeling set $V$\vspace{-1mm}}
    \label{fig:scaleV}
\end{figure}

To evaluate the scalability of RECEIPT, we measure its performance with $1, 2, 4, 9, 18$ and $36$ threads. Fig.\ref{fig:scaleU} and \ref{fig:scaleV} show the parallel speedup obtained by RECEIPT over single-threaded execution\footnote{We also developed a sequential version of RECEIPT with no synchronization primitives (atomics) and sequential implementations of basic kernels such as prefix scan, scatter etc. However, the observed performance difference between sequential implementation and single-threaded execution of parallel implementation was negligible.}. 

For most of the datasets, RECEIPT exhibits almost linear scalability. With $36$ threads, RECEIPT
achieves up to $17.1\times$ self-relative speedup. In comparison, we observed that \texttt{ParB}
exhibits at most $2.3\times$ parallel speedup (self-relative) with $36$ threads. RECEIPT's speedup is poor for \textit{ItV} because it is a very small dataset that gets peeled in $<4$s. It requires very little computation ($0.56$B wedges) and hence, employing a large number of threads is not useful. 

\kledit{Typically, datasets with small amount of wedges (\textit{ItV, LjV, EnV}) exhibit lower scalability,  because compared to larger datasets, they experience lower workload per synchronization round on average. For example, \textit{LjV} traverses $86\times$ fewer wedges than \textit{LjU} but incurs only $3.2\times$ fewer synchronizations. This increases the relative overheads of parallelization and restricts the parallel scalability of RECEIPT CD, which is the highest workload step in RECEIPT (fig.\ref{fig:wb}). For example, parallel speedup of RECEIPT CD with $36$ threads is $15.1\times$ for \textit{LjU} but only $7.1\times$ for \textit{LjV}.

In RECEIPT FD, parallel speedup is restricted by workload imbalance across the subgraphs. This is because RECEIPT CD tries to balance total wedge counts of vertex subsets as seen in original graph, whereas work done in RECEIPT FD is determined by wedges in induced subgraphs. Consequently, we observed that for some datasets, parallel scalability of RECEIPT FD is poorer than RECEIPT CD. For example, for \textit{TrU} with $36$ threads, parallel speedup of RECEIPT FD was only $5.3\times$ compared to $13.1\times$ of RECEIPT CD , $12.5\times$ of counting (pvBcnt) and $10.7\times$ of RECEIPT overall.
}

Note that even sequential RECEIPT is much faster than \texttt{BUP} because of the \looseness=-1following:
\begin{enumerate}[leftmargin=*]
    \itemsep0em
    \item \textit{Fewer support updates} -- updates to $\bowtie_u$ from all vertices in a peeling iteration are aggregated into a single \looseness=-1update.
    \item \textit{Lesser work} -- reduced wedge traversal due to HUC and DGM optimizations (sec.\ref{sec:optimizations}).
\end{enumerate}

We also observe that slope of the speedup curve decreases from $18$ to $36$ threads. This could potentially be due to the NUMA effects as RECEIPT does not currently have NUMA specific optimizations. Up to $18$ threads, the execution is limited to single socket but $36$ threads are spawned across two different sockets.

%% file: related.tex
\section{Related Work}\label{sec:related}
Dense subgraph discovery is a very crucial analytic used in several graph applications~\cite{anomalyDet, spamDet, communityDet, fang2020effective, otherapp1, otherapp2}. Researchers have developed a wide array of techniques for mining dense regions in unipartite graphs~\cite{sariyuce2016fast, fang2019efficient, gibson2005discovering, sariyuce2018local, angel2014dense, trussVLDB, lee2010survey, coreVLDB, coreVLDBJ, wang2018efficient}. 
%Parallel computing is often exploited to accelerate these methods and scale them to larger inputs~\cite{}. 
Among these, motif-based techniques have gained tremendous interest in the graph community~\cite{PMID:16873465, trussVLDB, tsourakakis2017scalable,tsourakakis2014novel, sariyucePeeling, aksoy2017measuring, wang2010triangulation}. Motifs like triangles represent a quantum of cohesion in graphs. Hence, the number of motifs incident on a vertex or an edge is an indicator of their involvement in dense subgraphs. Several recent works are focused on efficiently finding such motifs in the graphs~\cite{ahmed2015efficient, shiParbutterfly, wangButterfly, hu2018tricore, fox2018fast, ma2019linc}.

Nucleus decomposition is one such popular motif-based dense graph mining technique. It considers the distribution of motifs across vertices or edges within a subgraph as an indicator of its density~\cite{10.1007/s10115-016-0965-5}, resulting in better quality dense subgraphs compared to motif counting~\cite{10.1007/s10115-016-0965-5, sariyuce2015finding}. For example, truss decomposition mines subgraphs called $k$-trusses, where every edge participates in at least $k-2$ triangles within the subgraph.
%Thus, nucleus decomposition is able to find better quality dense subgraphs compared to simple motif counting. 
Truss decomposition also appears as one of the three tasks in the popular 
GraphChallenge~\cite{samsi2017static}, that has resulted in highly efficient parallel solutions scalable to 
billion edge graphs~\cite{date2017collaborative,voegele2017parallel,smith2017truss,green2017quickly}. 
However, such solutions cannot be directly applied on bipartite graphs due to the absence of triangles. 
Chakravarthy et al.\cite{chakaravarthy2018improved} propose a distributed truss decomposition algorithm that
trades off computational efficiency to reduce synchronization. This approach requires triangle 
enumeration and cannot be adopted for tip decomposition due to prohibitive space and computational 
requirements.

The simplest non-trivial motif in a bipartite graph is a Butterfly (2,2-biclique, quadrangle). Several algorithms covering various aspects of butterfly counting have been developed: in-memory or external memory~\cite{wangButterfly,wangRectangle}, exact or approximate counting~\cite{sanei2018butterfly,sanei2019fleet} and parallel counting on various platforms~\cite{shiParbutterfly, wangButterfly, wangRectangle}. Particularly, the in-memory algorithms for exact counting are relevant to our work. Wang et al.\cite{wangRectangle} count rectangles in bipartite graphs by traversing wedges with $\mathcal{O}\left(\sum_{u\in W}d_u^2\right)$ complexity. Sanei-Mehri et al.\cite{sanei2018butterfly} reduce this complexity to $\mathcal{O}\left(\min{\sum_{u\in U}d_u^2, \sum_{v\in V}d_v^2}\right)$ by choosing the vertex set where fewer wedges have end points.

Before the aforementioned works, Chiba and Nishizeki \cite{chibaArboricity} had proposed an $\mathcal{O}\left(\alpha\cdot m\right)$ complexity vertex-priority based quadrangle counting algorithm for generic graphs. Wang et al.\cite{wangButterfly} further propose a cache optimized variant of this algorithm and use shared-memory parallelism for acceleration. Independently, Shi et al.\cite{shiParbutterfly} develop provably efficient shared-memory parallel implementations of vertex-priority based counting. All of these approaches are amenable for per-vertex or per-edge \looseness=-1counting.

Butterfly based decomposition, albeit highly effective in finding quality dense regions in bipartite graphs, is computationally much more expensive than counting. Sariyuce et al.\cite{sariyucePeeling} defined $k$-tips and $k$-wings as subgraphs with minimum $k$ butterflies incident on every vertex and edge, respectively. Correspondingly, they defined the problems of Tip decomposition and Wing decomposition. Their algorithms use bottom-up peeling with respective complexities of $\mathcal{O}\left(\sum_{v\in V}d_v^2\right)$ and $\mathcal{O}\left(\sum_{(u,v)\in E}{\sum_{w\in N_v}{\left(d_u + d_w\right)}}\right)$. Independently, Zou~\cite{zouBitruss} defined the notion of bitruss similar to $k$-wing and showed its utility for bipartite network analysis. Shi et al.\cite{shiParbutterfly} propose parallel bottom-up peeling used as a baseline in our evaluation. Their wing decomposition uses hash tables to store edges and has a complexity of
$\mathcal{O}\left(\sum_{(u,v)\in E}{\sum_{w\in N_v}{\min{\left(d_u, d_w\right)}}}\right)$.

In their seminal paper, Chiba and Nishizeki \cite{chibaArboricity} had remarked that butterflies can be represented by storing the $\mathcal{O}(\alpha\cdot m)$ triples traversed during counting. In a very recent work, Wang et al.\cite{wangBitruss}  store these triples in the form of  maximal \textit{bloom} (biclique) structures, that enables quick retrieval of butterflies shared between edges. Based on this indexing, the authors develop the Bit-BU algorithm for peeling bipartite networks. To the best of our knowledge, it is the most efficient sequential algorithm that can perform wing decomposition in $\mathcal{O}\left(\bowtie_G\right)$ time. Yet, it takes more than $15$ hours to process the Livejournal~(\textit{Lj}) dataset whereas RECEIPT can tip decompose both \textit{LjU} and \textit{LjV} in less than $16$ minutes. Moreover, the bloom-based indexing has a non-trivial space requirement of $\mathcal{O}\left(\sum_{(u,v)\in E}{\min{\left(d_u, d_v\right)}}\right)$ which in practice, amounts to hundreds of gigabytes for datasets like Orkut(\textit{Or}). Comparatively, RECEIPT has a space-complexity of $\mathcal{O}\left(n\cdot T + m \right)$ (same as parallel counting \cite{wangButterfly}) and consumes only few gigabytes for all the datasets used in sec.\ref{sec:exp}. We also note that the fundamental ideas employed in RECEIPT and Bit-BU are complimentary. While RECEIPT attempts to exploit parallelism across $k$-tip hierarchy and reduce the problem  size for exact bottom-up peeling, Bit-BU tries to make peeling every edge more efficient. We believe that an amalgamation of our ideas with \cite{wangBitruss} can produce highly scalable parallel solutions for peeling large bipartite graphs.

%% file: conclusion.tex
\section{Extensions}\label{sec:extensions}
%\kldelete{Butterfly counting and peeling in bipartite networks is a fairly new problem. Research interests in this problem have grown tremendously in last few years and it still poses several unsolved challenges.}
In this section, we discuss opportunities for future research in the context of our work:
\begin{itemize}[leftmargin=*]
    \itemsep0em
    \item \textit{Parallel Edge peeling:} RECEIPT can be easily adapted for parallel wing decomposition (edge peeling) in bipartite graphs \cite{wangBitruss, sariyucePeeling}. The workload reduction optimizations in RECEIPT can have a greater impact on edge peeling due to the higher complexity and smaller range of wing numbers. However, there could be conflicts during parallel edge peeling as multiple edges in a butterfly could get deleted in the same iteration. Only one of the peeled edges should update the support of other edges in the butterfly, which can be achieved by imposing a priority ordering of edges. 
    %It will also be very interesting to develop and evaluate a combination of  RECEIPT with efficient (sequential) edge peeling algorithms proposed in \cite{wangButterfly}.
    
    \item \textit{Distributed Tip Decomposition:} Distributed-memory systems (like multi-node clusters) offer large amount of extendable computational resources and are widely used to scale high complexity graph analytics~\cite{chakaravarthy2018improved, bhattarai2019ceci,lakhotia13planting}.
    
    We believe that the fundamental concept of creating independent tip number ranges and vertex subsets can be very useful in exposing parallelism for distributed-memory algorithms. In the past, certain distributed algorithms for graph processing have achieved parallel scalability by creating independent parallel tasks \cite{lakhotia13planting, arifuzzaman2013patric}.
% \kldelete{However, inter-vertex communication of support updates should be optimized because of high network transfer costs~\cite{lumsdaine2007challenges,mcsherry2015scalability}. Further, even the relatively smaller amount of synchronization is RECEIPT (table~\ref{table:performance}) can significantly affect the performance of a distributed-memory system.}
    \kledit{However, support updates generated from peeling may need to be communicated on the network. This can affect scalability of the algorithm because of high communication cost in clusters~\cite{lumsdaine2007challenges, mcsherry2015scalability}. Further, execution on distributed systems may exhibit load imbalance due to larger number of threads and limitations of dynamic task scheduling across processes.}
    
    % Concept of splitting the tip-numbers into ranges and then decomposing each range 
    % individually can be used for distributed decomposition as well. That can scale the performance and 
    % extend our capabilities to even larger graphs.
    \item \textit{System Optimizations:} In this work, we have particularly focused on algorithmic optimizations. However, other aspects such as memory access optimizations~\cite{wei2016speedup,lakhotia2018accelerating} and SIMD parallelism~\cite{10.1145/3183713.3196924}  have significant impact on graph analytics. Enhancing memory access locality can also mitigate the NUMA effects that limit parallel speedup (sec.\ref{sec:scalability}).
    
    RECEIPT CD and RECEIPT FD exploit parallelism at vertex and subgraph granularity, respectively. Using fine-grained parallelism at edge or wedge granularity can further improve load balance.
    
    % \item Limitation - imbalance in fine-grained partitions still persists. Due to smaller
    % workload in this stage, its effect gets masked. Theroetical complexity analysis for the fine-grained
    % decomposition can help improve load balancing. Can be used to advantage in cloud computing,
    % use small number of machines/cores to do fine-decomposition, more economical without losing
    % much performance.
    
    % Overheads of parallelism, dynamic decision for number of partitions.
    % \item Improving the parallel scalability can provide a further 2x improvement in coarse-decomposition stage. Currently, limited by dynamic scheduling overheads in openmp.
    % \item Improving load balance between partitions can improve scalability of fine-decomposition stage.
    % Can be done by creating more partitions but increases overhead as seen in fig.\ref{}.
\end{itemize}

%section{Conclusion}
% In this paper, we propose a novel shared-memory parallel algorithm for tip decomposition -- RECEIPT. It is the first algorithm that exploits the massive parallelism across different levels of $k$-tip hierarchy. We also develop pragmatic optimizations to drastically improve the computational efficiency of RECEIPT. We empirically evaluate our approach on a $36-$core server and show that it provides orders of magnitude reduction in execution time, thread synchronizations and wedge traversal, compared to the state-of-the-art.

% We also explore the generalizability of RECEIPT for wing decomposition (edge peeling) and several interesting avenues for future work in this direction. We are currently investigating multi-level vertex partitioning to increase the parallelism and scalability of RECEIPT. We believe that scalable algorithms for parallel systems such as many-core CPUs, GPU or HPC clusters, can enhance the applicability of tip or wing decomposition, and our work is a step in that direction. 

\section{Conclusion}

In this paper, we proposed RECEIPT -- a novel shared-memory parallel algorithm for tip decomposition. RECEIPT is the first algorithm that exploits the massive parallelism across different levels of $k$-tip hierarchy. Further, we also developed pragmatic optimizations to drastically improve the computational efficiency of RECEIPT. 
% It employs a hybrid peeling strategy along with other optimizations, that considerably reduce the amount of wedge exploration and consequently its execution time.

We empirically evaluated our approach on a shared-memory multicore server, and showed that it can process some of the largest publicly available bipartite datasets orders of magnitude faster than the state-of-the-art algorithms, with dramatic reduction in synchronization and wedge traversal. Using $36$ threads, RECEIPT can provide up to $17.1\times$ self-relative speedup, being much more scalable then the best available parallel algorithms for tip decomposition.

We also explored the generalizability of RECEIPT for wing decomposition (edge peeling) and several interesting avenues for future work in this direction. 
% We are currently investigating multi-level vertex partitioning to increase the parallelism and scalability of RECEIPT. 
We believe that scalable algorithms for parallel systems such as many-core CPUs, GPU or HPC clusters, can enhance the applicability of tip or wing decomposition, and our work is a step in that direction.